\documentclass[a4paper,10pt]{article}
\usepackage{amsmath,amssymb,amsfonts,amsthm}
\usepackage{color}

\usepackage{graphicx}

\newtheorem{theorem}{Theorem}[section]
\newtheorem{lemma}[theorem]{Lemma}
\newtheorem{definition}[theorem]{Definition}

\newtheorem{Proposition}[theorem]{Proposition}

\newcommand{\chp}{\mathbb{H}^{2}_{\mathbb{C}}}

\newcommand{\vs}{\vspace{0.2cm}}

\newcommand{\n}{\noindent}
\newcommand{\ben}{\begin{equation*}}
\newcommand{\een}{\end{equation*}}

\newcommand{\be}{\begin{equation}}
\newcommand{\ee}{\end{equation}}
\newcommand{\bea}{\begin{eqnarray}}
\newcommand{\eea}{\end{eqnarray}}
\newcommand{\nn}{\nonumber}

\begin{document}
\title{Proof of the area-angular momentum-charge inequality for axisymmetric black holes} 
\author{Mar\'ia E. Gabach Clement\footnote{E-mail: gabach@aei.mpg.de} , Jos\'e Luis Jaramillo\footnote{E-mail: Jose-Luis.Jaramillo@aei.mpg.de} and Mart\'in Reiris\footnote{E-mail: martin.reiris@aei.mpg.de} \\
\\
Max Planck Institute for Gravitational Physics,\\ (Albert Einstein Institute), Am M\"uhlenberg 1,\\ D-14476 Golm, Germany.}
\date{}
\maketitle

\begin{abstract}
We give a comprehensive discussion, including a detailed proof, of the area-angular momentum-charge inequality for axisymmetric black holes. We analyze the inequality from several  viewpoints,  in particular
including aspects with a  
theoretical interest well beyond the Einstein-Maxwell theory. 
\end{abstract}

\vspace{0.3cm}

\begin{center}
\begin{minipage}[c]{11cm}
\tableofcontents
\end{minipage}
\end{center}
\newpage

\section{Introduction}

The main result of this article is 
the following theorem
\begin{theorem}\label{coromots}
Let $S$ be either 
\begin{enumerate}
\item[{\rm I.}] a smooth stable axisymmetric marginally outer trapped surface embedded in a spacetime, satisfying the dominant energy condition, with non negative cosmological constant $\Lambda$, angular momentum $J$, charges $Q_{\rm E}$ and $Q_{\rm M}$ and area $A$, or, 
\item[{\rm II.}] a smooth stable axisymmetric minimal surface in a maximal data set, satisfying the dominant energy condition, with non-negative $\Lambda$, angular momentum $J$, charges $Q_{\rm E}$ and $Q_{\rm M}$ and area $A$. 
\end{enumerate}
Then 
\begin{equation}\label{ine}
A^2\geq 16\pi^2[4J^2+(Q_{\mathrm E}^2+ Q_{\mathrm M}^2)^2].
\end{equation}
Moreover, the equality in \eqref{ine} is achieved if and only if the surface is the extreme Kerr-Newman sphere (see section \ref{seckerr}). 
\end{theorem}

This type of relation among physical parameters of black holes plays a relevant role in the context of the standard picture of classical gravitational collapse \cite{Pen73}. In this sense,  the works  of Penrose (see the review \cite{Mars09} on Penrose inequality) offer a paradigmatic example with the proposal of a lower bound for the total mass in terms of the size (area) of the black hole in the form $m^2\geq A/16\pi$. The efforts to formulate similar geometrical inequalities incorporating the angular momentum of the black hole have led to two different lines of research. The first one,  started in \cite{Dainmass08} and followed in \cite{ChruscielCosta09}, \cite{Costa09}, is of global nature and provides a lower bound to the total mass in terms of the angular momentum and charges in a vacuum black hole spacetime
\begin{equation}\label{MJQ}
 m^{2}\geq \frac{|J|^{2}}{m^{2}}+Q_{\rm E}^{2}+Q_{\rm M}^{2} \ .
\end{equation}
The second line of research leads to inequality (\ref{ine}), which  presents a quasilocal character
in the sense that only the geometry on a closed surface is involved in the analysis.

The first explicit lower bounds for the area solely in terms of black hole 
physical parameters, including the angular momentum, were given in \cite{HenAnsCed08,Hennig10,Ansorg:2010ru} 
(see also \cite{AnsorgPfister08}) in stationary black holes, 
and later in \cite{Acena11},  \cite{DainReiris11} and \cite{JaramilloReirisDain11} within dynamical scenarios 
(see also \cite{Gabach11}, \cite{DainJaramilloReiris11}, \cite{Simon:2011zf},
\cite{Hollands:2011sy}, \cite{GabachJaramillo11} and \cite{Jaramillo11pro} and 
the review article \cite{Dainreview11}  on the subject).

In the recent article \cite{GabachJaramillo11}, a first straightforward approach to prove inequality (\ref{ine}) in the dynamical case was presented. It consists in matching the variational problem discussed in \cite{Hennig10} for the stationary axisymmetric case  with the dynamical quasilocal treatment in \cite{JaramilloReirisDain11} (see also \cite{Chrusciel11} and \cite{Mars:2012sb} for further
clarification on the relation between the stationary and the dynamical quasilocal approaches). 
More specifically, as shown in \cite{GabachJaramillo11}, the proof of the strict case in point I of Theorem \ref{coromots} with vanishing magnetic charge $Q_{\mathrm M}=0$ follows directly from the proof in \cite{Hennig10} under the assumption of strict stability. We  note that the rigidity result is lacking in \cite{GabachJaramillo11}. We would like to mention that as this article was written, we have learned that the inclusion of the marginally stable case may also be done \cite{Ceder12} following the same procedure as in  \cite{Hennig10}, and whose resolution would lead to the full inequality (\ref{ine}).  

It is remarkable that both inequalities, \eqref{ine} and \eqref{MJQ} can be obtained via a variational 
principle involving energy \textit{flux} functionals \cite{Jaramillo11pro}. Although both procedures can be carried over without reference
to one another, the similitude between the functionals suggests a deeper relation between them,
and ultimately, a possible relation between the inequalities themselves.

In this article we pursue three goals. The first one is to establish and give a detailed proof of the AJQ inequality, completing and extending the analysis in \cite{GabachJaramillo11}. This is relevant for several reasons, namely, it gives information about the allowed values of the physical parameters for black holes. In particular it shows that, even in non-vacuum dynamical scenarios, the relations between these basic  parameters remain simple. Also, it puts in evidence the special role of extreme Kerr-Newman black hole, a fact that might shed light on the  stability of black holes. Finally, the relevance of this type of inequalities in the study of multiple black hole configurations, and as a powerful tool to probe known solutions was made clear in the work of Neugebauer et al. \cite{Neugebauer11} (see also \cite{Chrusciel11}), where they strongly use the uncharged version of \eqref{ine} to prove by contradiction, that two rotating black holes do not exist in equilibrium. 

The second goal of the present article is to gain insights about the underlying mechanisms leading to such an inequality as \eqref{ine}. In this respect, we expose two different approaches to the AJQ inequality. One of them relaxes to certain extent the axial symmetry assumption and makes use of harmonic maps between the surface and the complex hyperbolic space $\mathbb H^2_{\mathbb C}$. The second approach makes use of geodesics in $\mathbb H^2_{\mathbb C}$. Both approaches implement a minimization procedure which lead to Theorem \ref{coromots}, a procedure which seems to be needed due to the presence of the angular momentum (cf. \cite{DainJaramilloReiris11,Simon:2011zf}, where an inequality between area and electric and magnetic charges is obtained without axial symmetry and where no variational problem is formulated). 

The third goal of the article is to show how stable marginally outer trapped surfaces (MOTS's) and stable minimal surfaces over maximal surfaces can be treated on the same footing in the study of these quasilocal inequalities. Since the global characterization of a black hole in terms of notions such as the event horizon is of little practical use in the present quasilocal context, we must resort to quasilocal objects to represent black holes. Both, MOTS's and minimal surfaces has been extensively studied and used in the literature as signatures of the presence of a black hole region, at least in strongly predictable spacetimes \cite{HawEll73}, and more precisely, in the study of quasilocal inequalities, but, as far as we know, no link was established between the two types of surfaces in this context.  Although in the generic case there are fundamental differences between minimal surfaces and MOTS \cite{AnderssonMetzger09}, in this article we point out that their respective notions of stability crucially lead (in axisymmetry) to the same integral characterization and ultimately, to the same inequality.

Although much has been done during the last few years  in the field of geometrical inequalities for black holes, there are still many open questions to be studied. One of them is the possible explicit inclusion of the cosmological constant into the inequalities, in the presence of angular momentum (the
area-charge case has already been addressed in \cite{Simon:2011zf}). We emphasize that our result, Theorem \ref{coromots} allows the spacetime to have a non-negative $\Lambda$, but this quantity does not enter into the inequality \eqref{ine}. So we wonder how is inequality \eqref{ine} modified by its explicit introduction, and moreover, what happens with the negative Lambda case. Results 
in \cite{Simon:2011zf} provide a first step in this direction.

Another issue that must be better understood is the connection between the two types of inequalities mentioned above, \eqref{ine} and \eqref{MJQ}. We give some insights in this article (see the appendix), but there are many issues that are not entirely yet clear. This is not an easy problem, since it involves linking global and quasilocal viewpoints. It would be, however, very desirable, since its full resolution would give a concrete probe to compare with the Penrose inequality.

Finally, we want to mention that this type of quasilocal inequalities has been discussed in a broader context lately, mainly by Dain, \cite{Dainreview11}, \cite{Dainpersonal}, and we are forced to wonder about the universal validity of such a relation. Within the context of electrovacuum black holes, in this article we give a first step by studying the case of general surfaces within maximal initial data (that is, surfaces that are not necessarily minimal), and prove its validity. We understand that there is much work to do in order to generalize the results presented here to ordinary objects. Nevertheless, due to the rigidity statement in Theorem \ref{coromots}, and to the special properties of black holes in nature, one might expect that the extreme Kerr-Newman sphere should play a key role also in the general setting. We believe this will be an active field of research during the next years.

The article is organized as follows. In section \ref{settings} we introduce the basic 
elements needed for the statement of our main result. This includes the formal definitions of angular momentum and charges of a surface within the Einstein-Maxwell-matter theory and an outline of stable axially symmetric marginally outer trapped surfaces and stable axially symmetric minimal surfaces over maximal slices. In particular, as we mentioned above, we will show that the stability condition for both surfaces leads to the same integral characterization. Finally but crucially, 
we identify a set of suitable potentials to describe the gravitational and electromagnetic
 fields which proves to be useful for handling the variational problem needed to establish inequality \eqref{ine}.

In section \ref{main} we present the main partial results leading to Theorem \ref{coromots}, which are written up in the form of three Lemmas, \ref{lema1}, \ref{teo}, and \ref{unic} and give, respectively, a lower bound to the area in terms of a functional on the 2-sphere, an absolute lower bound to this functional, and the rigidity statement. Moreover, in section \ref{application}, we present an interesting application  to black hole initial data which intends to study the general validity of the AJQ inequality for black hole spacetimes. In section \ref{seckerr} we  study the so-called extreme Kerr-Newman sphere, pointing out the MOTS and minimal surface viewpoints and its connection. We also give an 
interesting geometric description of the extreme Kerr-Newman horizon geometry in terms of semicircles in the complex hyperbolic space. 

In section \ref{avenidas} we present the proof of Theorem \ref{teo}. We do so by following two approaches and highlighting different aspects of the underlying structure. The first one, in section \ref{prueba1} makes contact with harmonic maps, and the second one, in section \ref{prueba2} solves the minimization problem by identifying the minimizers with geodesics in the complex hyperbolic space.

We also include an appendix where we discuss the possible relation between quasilocal and global inequalities. 

\section{Settings}\label{settings}

In this section we introduce the objects that will be used as the black hole signatures, namely stable marginally outer trapped surfaces and stable minimal surfaces. We will expose their main properties and, more importantly, we will show how, under the axisymmetry hypothesis, the stability notions for both types of surfaces lead to a single inequality from which \eqref{ine} is obtained. In order to do so, we begin with a brief outline of closed surfaces embedded in a spacetime, their intrinsic and extrinsic geometry and the physical quantities one can associate to them. 

Let ($\mathcal V, g_{ab}$) be a spacetime satisfying the Einstein equations
\begin{equation}
  \label{eq:3a}
  G_{ab}=8\pi (T^{EM}_{ab}+T^M_{ab})-\Lambda g_{ab} \ ,
\end{equation}
where $G_{ab}:=R_{ab}-\frac{1}{2}Rg_{ab}$ is the Einstein tensor,  $g_{ab}$ and $\nabla_a$ are the spacetime metric and its Levi-Civita connection respectively, $ \Lambda\geq0$ is a non-negative cosmological constant and
 we have decomposed the stress-energy tensor $T_{ab}$ into  its electromagnetic $T^{EM}_{ab}$ and non-electromagnetic $T^M_{ab}$ components. We assume that the latter satisfies the dominant energy condition.

Consider a spacelike surface $S$ embedded in the spacetime, with induced metric $q_{ab}$ and Levi-Civita 
connection $D_a$. Let $\ell^a$ and $k^a$ be future-oriented null vectors normal to $S$ such that $\ell_ak^a=-1$ and 
$\ell^a$ is outward-pointing. Regarding the extrinsic curvature elements, we introduce the expansion associated to  $\ell^a$, $\theta^{(\ell)}:=q^{ab}\nabla_a\ell_b$, the shear tensor $\sigma_{ab}^{(\ell)}:=q^c_aq^d_b\nabla_c\ell_d-\frac{1}{2}\theta^{(\ell)}q_{ab}$, and the normal fundamental form $\Omega_a^{(\ell)}:=-k^c{q^d}_a\nabla_d\ell_c$. It is important to remark that the normalization condition on the null normals $\ell^a$, $k^a$ leaves a boost rescaling freedom: $\ell'^a=f \ell^a, k'^a=f^{-1} k^a$ under which $\theta^{(\ell)}$ and $\Omega_a^{(\ell)}$ transform respectively as $\theta^{(\ell')}=f\theta^{(\ell)}$ and  $\Omega_a^{(\ell')}= \Omega_a^{(\ell)}+D_a\ln f$.

Although the main inequality can be understood more naturally in the context of globally axisymmetric black hole space-times,  it is remarkable that in fact, only very little quasilocal (rather than global) axisymmetry is necessary for its validity. For this reason we give here the most basic notion of {axisymmetry} under which (\ref{ine}) is valid.

We say that the closed surface $S$ is {\it axisymmetric} if there exists 
a  Killing vector field $\eta^{a}$ on ${\cal S}$, i.e.
${\mathcal L}_{\eta}q_{ab}=0$,   with closed integral curves and 
normalized so that its integral curves have an affine length of $2\pi$, 
and such that
\begin{equation}
{\mathcal{L}}_{\eta} \Omega_a^{(\ell)} = {\mathcal{L}}_{\eta} \Pi(A_a) = {\mathcal{L}}_{\eta} E_{\perp}={\mathcal{L}}_{\eta} B_{\perp}=0.
\end{equation}
Above $A_{a}$ is the electromagnetic potential  given by $F_{ab}=\nabla_aA_b-\nabla_bA_a$,  $F_{ab}$ is the electromagnetic field tensor, $\Pi(A_{a})$ is the pullback of the form $A_{a}$ to the tangent space of $S$ and $E_{\perp}$ and $B_{\perp}$ are the electric and magnetic fluxes 
across $S$, given by
\begin{equation}\label{emfields}
E_\perp:=\ell^a k^b F_{ab} \ ,\qquad B_\perp:=\ell^a k^b {}^*\!F_{ab} \ ,
\end{equation}
 where ${}^*\!F_{ab}$ is the dual of $F_{ab}$. Note that $E_\perp$ and $B_\perp$ are independent of a conformal rescaling
of the null normals. 

If $S$ is axisymmetric, then   we define the  projection of $\Omega_a^{(\ell)}$ along the Killing vector $\eta^a$ , $\Omega^{(\eta)}_a:=\eta^b\Omega^{(\ell)}_b\eta_a/\eta$, where $\eta:=\eta^a \eta_a$. Crucially, 
$\Omega^{(\eta)}_a$ is then divergence-free and therefore invariant under 
null normal rescalings preserving the axisymmetry.

\subsection{Angular momentum and electromagnetic charges}\label{angular}
We now introduce three physical quantities\footnote{Note that the sign convention in this article
is consistent with that in \cite{AshFaiKri00}, \cite{Weinstein96}  and opposite to that
in \cite{GabachJaramillo11}, \cite{Booth08}. 
This does not affect the inequality \eqref{ine}, that involves only quadratic expressions.
}  
associated with a surface $S$ 
in the context of Einstein-Maxwell-matter theory, namely, the charges and the angular momentum.
  
 Following \cite{AshFaiKri00} we write the electric and magnetic charges of the surface $S$,
 respectively,  as 
\begin{equation}\label{cargacampos}
Q_{\mathrm E} = Q_{\mathrm E}(S) := -\frac{1}{4\pi}\int_S E_\perp dS \ ,\qquad Q_{\mathrm M}  =
 Q_{\mathrm M}(S):=-\frac{1}{4\pi}\int_S B_\perp dS \ ,
\end{equation}
where $dS$ is the area element of $S$. By integrating Maxwell's equations $j^a=\nabla_bF^{ba}$ and $0=\nabla_b{}^*\!F^{ba}$, where $j^a$ is the electric charge current, we have the conservation  law 
\begin{equation}
Q_{\mathrm E}(\partial \Sigma)= \int j^a n_adV,\ \  Q_{\mathrm M}(\partial \Sigma)= 0 \ ,
\end{equation}
where $\Sigma$ is a spatial 3-slice with boundary $\partial \Sigma$,  $n^b$ is the unit normal vector to $\Sigma$ and $dV$ is the volume element in $\Sigma$.  In particular, this shows that in the absence of matter between two surfaces $S$, $S'$ the charges are conserved, i.e. $Q(S)=Q(S')$.

If the surface $S$ is axially symmetric with axial Killing vector $\eta^a$, then one can define \cite{Ashtekar:2001is,Booth08} a
\textit{canonical angular
 momentum}  of $S$ given, within the Einstein-Maxwell-matter context, by

\begin{equation}\label{j}
J = J(S):=- \frac{1}{8\pi}\int_S\Omega_a^{(\ell)}\eta^a dS - \frac{1}{4\pi}\int_S A_a\eta^aE_\perp dS.
\end{equation}
If the axial vector $\eta^a$ is the restriction of a global spacetime axisymmetric vector,
 $J$ can be expressed as \cite{Carter72} 
\begin{equation}\label{j_Komar} 
J(S) = - \frac{1}{16\pi} \int_{\cal S} \nabla^b\eta^a dS_{ab} - \frac{1}{4\pi}\int_S A_a\eta^aE_\perp dS \ ,
\end{equation}
with $dS_{ab} = (\ell_a k_b - k_a \ell_b) dS$. Note that, as given by \eqref{j_Komar} $J(S)$ is well defined for an
arbitrary non-necessarily axisymmetric surface $S$. The first term is the so-called
Komar angular momentum
 $J_K$, i.e. $J_K := - \frac{1}{16\pi} \int_{\cal S} \nabla^b\eta^a dS_{ab}$.
Moreover,  we also have
\cite{Carter72,gourgoulhon20123+1}
\bea
\label{e:conservation}
J(\partial\Sigma) = -\int_\Sigma T^M_{ab} \eta^an^bdV  -\int_\Sigma \eta^a A_a j_b n^bdV .
\eea
Therefore, in the absence of matter between  surfaces $S$ and $S'$, the angular
momentum (\ref{j_Komar}) is conserved, $J(S)=J(S')$.

\subsection{Stable MOTS's and stable minimal surfaces}\label{surfaces}

We recall here the definitions of stable marginally outer trapped surfaces in a given spacetime
and stable minimal surfaces over maximal slices, which are the main two types of surfaces we are interested in this article. 

We say that $S$ is a marginally outer trapped surface, or MOTS if $\theta^{(\ell)}=0$. 
Moreover, we  say that it is stable 
(or more precisely, spacetime stably outermost, 
according to the definition in \cite{JaramilloReirisDain11}; see also 
\cite{AndMarSim05,AndMarSim08,Hay94,Racz:2008tf}) if there exists 
an outgoing vector $X^a= \gamma \ell^a - \psi k^a$, with functions
$\gamma\geq0$, $\psi>0$, such that $\delta_X\theta^{(\ell)}\geq0$. Here $\delta_X$ 
denotes the deformation operator on $S$ \cite{AndMarSim05,AndMarSim08,Booth08} that controls the infinitesimal variations of geometric objects defined on $S$ under an infinitesimal deformation of the surface along the vector $X^a$.
$S$ is a {\em stable axisymmetric MOTS} if $S$ is axisymmetric and stable with 
axisymmetric $\gamma, \psi$. Given a stable axisymmetric MOTS $S$ with axial Killing vector $\eta^a$, the stability condition for $S$ is translated into the inequality (see \cite{JaramilloReirisDain11,Jaramillo:2012zi} for details)
\bea
\label{e:inequality_alpha}
\int_{\cal S} \left[|D\alpha|_q^2 + \frac{R_S}{2} \alpha^2  
 \right] dS \geq
 \int_{\cal S} \left[  |\Omega^{(\eta)}|_q^2  \alpha^2 +
 |\sigma^{(\ell)}|_q^2  \alpha \beta
+ G_{ab}\alpha\ell^a (\alpha k^b + \beta\ell^b) \right] dS \ ,
\eea  
valid for all  axisymmetric functions $\alpha$ on $S$.
Here $|\,\cdot\,|_q$ is the norm with respect to the  2-metric $q_{ab}$ on $S$, $R_S$ is the scalar curvature on $S$, and $\beta:=\alpha\gamma/\psi$. Use the Einstein equations \eqref{eq:3a} and insert the expression  $8\pi T^{EM}_{ab}\ell^ak^b = E_\perp^2+B_\perp^2$ (see  \cite{Booth08,DainJaramilloReiris11}
in \eqref{e:inequality_alpha} to obtain the inequality (see \cite{GabachJaramillo11})
\be
\label{e:inequality_alpha_AJQ}
\int_{S} \left[|D \alpha|_q^{2}+\frac{R_{S}}{2} \alpha^{2}\right] dS\geq \int_{S} \left[|\Omega^{(\eta)}|_q^{2} + E_{\perp}^{2}+B_{\perp}^{2}\right]\alpha^{2}dS.
\ee
To  pass from the stability condition \eqref{e:inequality_alpha} to inequality \eqref{e:inequality_alpha_AJQ} we have also discarded the 
non-negative shear term, the non-electromagnetic matter contribution (due to the energy condition) and the non-negative cosmological constant.

\vs 
We introduce now the second type of surfaces we are interested in this article, namely, stable axisymmetric minimal surfaces over  maximal slices. Consider a maximal initial data $(\Sigma, h_{ab},K_{ab},E_a,B_a)$ for the Einstein-Maxwell-matter system, where $h_{ab}$, $K_{ab}$ are the first and second fundamental forms of $\Sigma$ respectively, and $E_a:=F_{ab}n^b$, $B_a:={}^*\!F_{ab}n^b$ are the electromagnetic fields on $\Sigma$. As the datum is maximal we have $h^{ab}K_{ab}=0$. Suppose that $S\subset \Sigma$ is a minimal surface, namely one whose mean curvature (inside $\Sigma)$ is zero. The surface $S$ is stable if the second variation of the 
area is non negative, $\delta^2_{\alpha e_1}A\geq0$ for all functions $\alpha$. Suppose now that $S$ is an axisymmetric surface in the sense introduced before where $\ell=n+e_{1}$, $k=n-e_{1}$  and where $n$ is the space-time-normal to $\Sigma$ and $e_{1}$ one of the two normals to $S$ in $\Sigma$. Then the axisymmetric surface $S$ is stable if the second variation of the 
area is non negative, $\delta^2_{\alpha e_1}A\geq0$ for all axisymmetric functions $\alpha$. It is worth noting, to conciliate the stability definition for minimal surfaces and MOTS, that if for an axisymmetric minimal surface there is $\gamma>0$ such that $\delta_{\gamma e_{1}}\theta\geq0$ (where $\theta$ are the mean curvatures in $\Sigma$) then the surface is stable in the sense before.

Given a stable axisymmetric minimal surface $S$ in a maximal slice, the stability condition is translated into the standard
\be
\int_{S} \left[|D\alpha|^{2}+\frac{R_{S}}{2}\alpha^{2}\right] dS\geq \int_{S} \frac{1}{2}\left[R +|\hat{\Theta}|^{2}\right]\alpha^{2} dS
\ee
where $\hat{\Theta}$ is the trace-less part of the second fundamental form $\Theta$ and $R$ is the scalar curvature of the slice. Using the energy constraint $R=|K|^{2}+2(|E|^{2}+|B|^{2})+16\pi T^{M}_{ab}n^{a}n^{b}$ and that $\Omega^{(\eta)}_{a}\eta^{a}=-K(\eta,e_{1})$ we obtain after discarding some non-negative quadratic terms (in $|K|^{2}$ and $|E|^{2}+|B|^{2}$) and the non-negative term in $T^{M}_{ab}n^{a}n^{b}$, exactly the same 
inequality \eqref{e:inequality_alpha_AJQ} which was obtained for stable MOTS's.

\subsection{The quasilocal potentials}\label{potentials}

In this section we write the relevant components of the  intrinsic and extrinsic geometry of $S$, together with the electromagnetic field in terms of
 a set of potentials $\mathcal D=(\sigma,\omega,\psi,\chi)$ which are appropriate for applying the variational
 procedure which proves inequality \eqref{ine}.

Let $S$ be either an axisymmetric stable MOTS, or an axisymmetric stable
 minimal surface (in a maximal slice). We assume that over $S$ either $J$, $Q_{\rm E}$ or $Q_{\rm M}$ are non-zero, otherwise there
 is nothing to prove and the inequality \eqref{ine} is trivial. Choosing $\alpha=1$ in (\ref{e:inequality_alpha_AJQ}), and applying the  Gauss-Bonnet theorem, it follows that the Euler characteristic of $S$ is positive and therefore $S$ 
is topologically a sphere. Thus the metric over $S$ can be written uniquely in the 
form (see \cite{AshEngPaw04,DainReiris11})
\begin{equation}
\label{SDE}
ds^2= e^{2c-\sigma} d\theta^2 + 
e^\sigma\mathrm{sin}^2\theta d\varphi^2\ ,
\end{equation}
where $c$ is a constant.
With this choice of coordinate system, the area element and area of $S$ are given respectively by
$dS=e^c dS_0$, with $dS_0= \mathrm{sin}\theta d\theta d\varphi$, and $A= 4\pi e^c$. 
Moreover, the regularity of the metric  at the poles requires ${\sigma}|_{\theta=0,\pi}=c$.
In addition, the squared norm $\eta$ of the axial Killing vector 
$\eta^a$ is given by $\eta = e^\sigma \mathrm{sin}^2\theta$.

Regarding the 1-form $\Omega^{(\ell)}_a$, we write its Hodge decomposition in divergence-free
and exact parts (see \cite{JaramilloReirisDain11,Jaramillo11pro})
\begin{equation}\label{e:Omega_ell}
 \Omega^{(\ell)}_a =  \epsilon_{ab}D^b \tilde{\omega} + D_a \lambda \ ,
\end{equation}
for some regular functions $\tilde\omega$ and $\lambda$ on $S$. 
From the axisymmetry of $S$ it follows that this divergence-free part 
is given by $\Omega^{(\eta)}_a$. Explicitly
\bea\label{deftildeomega}
\Omega^{(\eta)}_\theta= 0\ ,   \qquad \Omega^{(\eta)}_\varphi = -e^{\sigma-c}\sin\theta\;\tilde\omega' \ , 
\eea
where the prime denotes derivative with respect to the variable $\theta$.
Next, let $\psi,\chi, \omega$ be regular functions of $\theta$ defined through the following expressions,
\begin{equation}\label{defpsichi}
\psi'= -E_\perp e^c\sin\theta,\qquad \chi'= -B_\perp e^c\sin\theta,
\end{equation}
\begin{equation}\label{defomega}
\omega'= 2\eta\tilde\omega' -2\chi\psi' +2\psi\chi'.
\end{equation}
It is remarkable that with this choice of potentials, the charges and angular momentum are given by the boundary
 values of $\psi$,  $\chi$ and $\omega$. To see this, use \eqref{cargacampos} and \eqref{defpsichi} to get 
\begin{equation}\label{condpotenciales1}
Q_{\mathrm E}=\frac{\psi(\pi)-\psi(0)}{2},\qquad Q_{\mathrm M}=\frac{\chi(\pi)-\chi(0)}{2}
\end{equation}
and use \eqref{j} and \eqref{defomega} to find
\begin{equation}\label{condpotenciales2}
J=\frac{\omega(\pi)-\omega(0)}{8}.
\end{equation}
Moreover, since the potentials $\psi$, $\chi$, $\omega$ are defined up to an additive constant, we assume, 
without lost of generality, that $\psi(\pi)=-\psi(0)=Q_{\mathrm E}$, $\chi(\pi)=-\chi(0)=Q_{\mathrm M}$ and
 $\omega(\pi)=-\omega(0)=4J$.

Finally, note that the function $\omega_K$ defined through $\omega_K' = 2\eta \tilde{\omega}'$ provides a potential for the Komar angular
momentum $J_K$ (cf. \cite{JaramilloReirisDain11}): $\omega_K(\pi)=-\omega_K(0)=4J_K$.

\section{Discussion of the main results}\label{main}

 We present here the main results leading to the AJQ inequality, a discussion about its more general validity and a detailed study of the unique minimizer for the area, namely, the extreme Kerr-Newman sphere. 

We begin by  stating three lemmas which, together, prove Theorem \ref{coromots}. Lemma \ref{lema1} establishes a lower bound for the area in terms of a bounded functional $\mathcal M$. This result comes up simply by rewriting the stability condition \eqref{e:inequality_alpha_AJQ} for the surface $S$ in terms of the set $\mathcal D$. The second statement, Lemma \ref{teo} gives an explicit sharp bound for the functional $\mathcal M$ in terms of the angular momentum and charges. These two Lemmas prove the AJQ inequality. The final statement, Lemma \ref{unic} proves that the AJQ inequality is saturated by a unique set $\mathcal D_0$  called extreme Kerr-Newman sphere. We will discuss this special set in section \ref{seckerr}.

Consider the stability condition \eqref{e:inequality_alpha_AJQ} valid for axisymmetric MOTS's and minimal surfaces (over maximal slices). Since it holds for any axisymmetric function $\alpha$, we take, as in \cite{DainJaramilloReiris11}, \cite{JaramilloReirisDain11}, \cite{GabachJaramillo11} the probe function
\bea
\label{e:alpha}
\alpha=e^{c-\sigma/2}.
\eea
Some insights about the reason for this choice of $\alpha$ is provided in section \ref{seckerr}. Then rewrite \eqref{e:inequality_alpha_AJQ} in terms of the potentials (\ref{defpsichi})-(\ref{defomega}) and use $A=4\pi e^c$  to arrive, following \cite{DainReiris11,JaramilloReirisDain11,GabachJaramillo11}, at a fundamental inequality  which is summarized in the following lemma \cite{GabachJaramillo11}.

\begin{lemma}\label{lema1}
Let $S$ be an axisymmetric stable MOTS or an axisymmetric stable minimal surface in  a maximal slice. Then 
\begin{equation}\label{ineqinicial}
A\geq 4\pi e^{\frac{\mathcal M-8}{8}},
\end{equation}
where $\mathcal M$ is given by 
\begin{equation}\label{defM}
\mathcal M:=\frac{1}{2\pi}\int\left[4\sigma+|D\sigma|^2+\frac{|D\omega+2\chi D\psi-2\psi D\chi|^2}{\eta^2}+4\frac{|D\psi|^2+|D\chi|^2}{\eta}\right]dS_0,
\end{equation}
and the norm $|\,\cdot\,|$ is taken with respect to the standard round metric on $S^2$.
\end{lemma}

A fundamental sharp lower bound for the functional ${\mathcal{M}}$ is stated in the following Lemma.
\begin{lemma}\label{teo}
Let $\mathcal D=(\sigma,\omega,\psi,\chi)$ be a regular set on $S^2$ with fixed values of $J$, $Q_{\rm E}$ and $Q_{\rm M}$. Then 
\begin{equation}\label{Mine}
e^{\frac{{\mathcal{M}}-8}{4}}\geq\ 4J^{2} + Q^{4},
\end{equation}
with $Q^2=Q_{\mathrm{E}}^2 + Q_{\mathrm{M}}^2$.
\end{lemma}

The proof of this result involves a minimization problem and can be approached in different ways, that will be discussed in full detail in section \ref{avenidas}. 

We want to emphasize that Lemma \ref{teo} does not assume axisymmetry on the set $\mathcal D=(\sigma,\omega,\psi,\chi)$. One of the proofs will deal with these non-necessarily axisymmetric sets (see section \ref{prueba1}). Finally, we give the precise definition of regular set mentioned in the Lemma. 

\vs

\begin{definition}\label{regularset} The set $\mathcal D=(\sigma,\omega,\psi,\chi)$ on $S^2$ is a regular set if the functions $\sigma$, $\omega$, $\psi$ and $\chi$ are $C^\infty$ on $S^2$, and moreover, we have the following behavior near the poles
\begin{enumerate}
\item[\rm {i)}] $\omega=\pm 4|J|+O(\sin^2\theta),\quad \psi=\pm Q_{\rm E}+O(\sin^2\theta),\quad \chi=\pm Q_{\rm M}+O(\sin^2\theta)$,

where the signs $+,-$ refer to the values at $\theta=\pi,0$ respectively.
\item[\rm {ii)}] $|D\omega+2\chi D\psi-2\psi D\chi|=O(\sin^ 3\theta)$.
\end{enumerate}
\end{definition}
We remark that if the functions $\omega,\psi,\chi$ arise from a smooth set of axisymmetric fields $\Omega^{(\eta)}_a$, $E_\perp$, $B_\perp$ via equations \eqref{deftildeomega}, \eqref{defpsichi}, \eqref{defomega}, then, they satisfy items i) and ii) of Definition \ref{regularset} automatically.

It is also important to stress that the Lemma \ref{teo} is also valid for smooth functions $\sigma, \omega, \psi, \chi$ such that they satisfy condition i) in the above definition and $\mathcal M$ is finite, condition ii) being no longer necessary. We will come back to this point in section \ref{prueba1}.

\vs 
The final result we present concerns the uniqueness of a regular set saturating inequality \eqref{ine}.
\begin{lemma}\label{unic}
There exists a unique regular set $\mathcal D$ saturating the AJQ inequality (\ref{ine}), with $A=4\pi e^{\sigma}\big|_{\theta=0,\pi}$ and it is the extreme Kerr-Newman sphere set $\mathcal D_0=(\sigma_0,\omega_0,\psi_0,\chi_0)$  given by
\begin{align}
\label{EKN1} \sigma_0&=\ln \frac{\left(2a_{0}^2+Q^2\right)^2}{\Sigma_0},\\
\label{EKN2} \omega_0&=-4J\frac{2a_{0}^ 2+Q^2}{\Sigma_0}\cos\theta,\\
\label{EKN3} \psi_0 &=-\frac{Q_{\mathrm E}(2a_{0}^2+Q^2)\cos\theta-Q_{\mathrm M} a_{0}\sqrt{a_{0}^2+Q^2}\sin^2\theta}{\Sigma_0},\\
\label{EKN4} \chi_{0}&=-\frac{Q_{\mathrm M}(2a_{0}^2+Q^2)\cos\theta-Q_{\mathrm E} a_{0}\sqrt{a_{0}^2+Q^2}\sin^2\theta }{\Sigma_0} 
\end{align}
with 
\begin{gather}
a_{0}=\frac{J}{m_{0}},\ \ m_{0}=\sqrt{\frac{Q^{2}+\sqrt{4J^{2}+Q^{4}}}{2}},\\
\Sigma_0 = Q^ 2+a_{0}^ 2(1+\cos^ 2\theta)
\end{gather}
\end{lemma}
In Section \ref{seckerr} we will discuss the properties of the minimizer set $\mathcal D_0$ and show that this special datum appears in two non-equivalent important and concrete contexts, 
\begin{enumerate}
\item[\small\bf O1.] on a MOTS in the horizon of the extreme Kerr-Newman solution, and, 
\item[\small \bf O2.]  on a  minimal sphere in the {\it extreme Kerr-Newman throat}, which is a maximal initial datum.
 \end{enumerate}

\subsection{On the general validity of the AJQ inequality }\label{application}

The case II in Theorem \ref{coromots} allows to show  that, in some situations, the $AJQ$-inequality \eqref{ine} is valid for any surface and not just for stable minimal axisymmetric surfaces over a maximal slice. The particular situations in consideration will be those of ``trumpet" and ``doubly asymptotically flat" axisymmetric initial data.

This result is interesting in the light of the conjecture that the foliation by maximal slices, whose leaves are all of the same type (trumpet or doubly AF), is believed to cover the whole domain of outer communication of the black holes. For this type of solutions one expects the inequality to hold over a large variety of surfaces in the whole domain of outer communication. It is worth stressing that one does not expect the equality in \eqref{ine} to be achieved at any surface in the trumpet or doubly AF maximal slice \cite{Reiris12}.

To define axisymmetric ``trumpet" initial data sets we follow \cite{WaxeneggerBeigMurchadha} and refer the reader to this article for more details. A ``trumpet" initial data set for the Einstein-Maxwell equations is a maximal and axisymmetric electrovacuum initial datum ($\Sigma; h, K; E,B$), with $\Sigma\approx \mathbb{R}^{3}\setminus \{0\}$ and $\Sigma/_{U(1)}\approx [0,1]\times \mathbb{R}$ (in particular with an axis having two connected components) and with particular asymptotics at the origin and at infinity. Precisely, we require asymptotic flatness at infinity (of $\mathbb{R}^{3}$), and at the origin (of $\mathbb{R}^{3}$) requiring $h$ to approach a cylindrical metric in the following sense: there exists a diffeomorphism $\Phi$ between, say, $B_{1/2}\setminus\{0\}$ and $(T,\infty)\times S^2$ so that $(\Phi_*\bar h)_{ij}-\bar h_{ij}=o(1)$ as $t\to\infty$, where $\bar h$ denotes the cylindrical metric of the form $\bar h=f^{2}dt^2+q$, with $q$ a Riemannian metric on $S^2$. We refer to the origin as a cylindrical end. ``Doubly asymptotically flat" initial data are defined in exactly the same way but now with two asymptotically flat ends, at infinity and at the origin of $\mathbb{R}^{3}$. 
We prove below the following proposition.

\begin{Proposition} \label{trumpet} Consider either an axisymmetric ``doubly asymptotically flat" or a ``trumpet" maximal initial datum $(\Sigma;(g,K);(E,B))$ for the electrovacuum system, with total angular momentum and charges ${\mathcal J}$, ${\mathcal Q}_{\rm E}$ and ${\mathcal Q}_{\rm M}$. Then for any oriented, non-necessarily axisymmetric embedded surface $S$ of arbitrary topology, its angular momentum and charges are given by one of the following two possibilities 
\begin{align}
\label{VFC} J=0,\ Q_{\rm E}=0,\ Q_{\rm M}=0,\ \ or,\\
\label{VSC} J={\mathcal J},\ Q_{\rm E}={\mathcal Q}_{\rm E},\ Q_{\rm M}={\mathcal Q}_{\rm M}
\end{align}
Moreover the $AJQ$-inequality (\ref{ine}) holds.
\end{Proposition}  
\begin{proof} 
\vs
To better visualize the proof, let us assume that we choose a diffeomorphism between $\mathbb{R}^{3}\setminus \{0\}$ and $\Sigma$ in such a way that the orbits of the Killing field, as seen in $\mathbb{R}^{3}\setminus \{0\}$, are exactly those circles which are the rotations of points around the $z$-axis. In this way the two components of the axis  are given by $\{(x,y,z),x=y=0,z>0\}$ and $\{(x,y,z),x=y=0,z<0\}$.

Let $S$ be an oriented surface. As a surface in $\mathbb{R}^{3}\setminus \{0\}\subset \mathbb{R}^{3}$ it divides $\mathbb{R}^{3}$ into two connected components. If the unbounded component contains the origin $\{0\}$ then $S$ encloses (including $S$) a compact region in $\mathbb{R}^{3}\setminus \{0\}$ and therefore (by Gauss theorem) $J$, $Q_{\rm E}$ and $Q_{\rm M}$ are zero, namely their values are as in (\ref{VFC}). In this case (\ref{ine}) is trivial. We assume therefore that it is the bounded component that contains the origin. In this case the values of $J,Q_{\rm E}$ and $Q_{\rm M}$ are (by Gauss theorem again) those of the end, namely as in (\ref{VSC}). 

Now,  in order to prove that the AJQ inequality \eqref{ine} is satisfied, assume by contradiction that (\ref{ine}) does not hold. Following \cite{MeeksSimonYau}, there are surfaces\footnote{The conclusion is direct for ``doubly asymptotically flat" initial data. For ``trumpet" data it requires a little more effort but feasible by taking into account that the ``asymptotic spheres" over the cylindrical end satisfy (\ref{ine}).} (possibly repeated) $S_{1},\ldots,S_{m}$ realizing the infimum of the the areas $A(\tilde{S})$ where $\tilde{S}$ is isotopic to $S$, namely $\sum A(S_{i})=\inf\{A(\tilde{S}),\tilde{S}\sim S\}$, where $\tilde{S}\sim S$ means that $\tilde{S}$ is isotopic to $S$ (one can see that the infimum is non-zero). Moreover, the surfaces are non-contractible (to a point) in $\mathbb{R}^{3}\setminus \{0\}$ and are also embedded. It follows that they are orientable (otherwise are contractible) and stable. As the manifold $(\Sigma,g)$ is axisymmetric (complete) and non-compact then every $S_{i}$, $i=1,\ldots,m$ is known to be axisymmetric. So each of them is either an axisymmetric sphere or an axisymmetric torus (there are no axisymmetric surfaces of higher genus). But, any axisymmetric torus is contractible to a point in $\mathbb{R}^{3}\setminus \{0\}$, which is not possible. Therefore all the $S_{i}'s$ are axisymmetric spheres and as they are non-contractible (to a point) in $\mathbb{R}^{3}\setminus \{0\}$ then they all must enclose the origin. Thus, the angular momentum and charges of, say, $S_{1}$, are the given ${\mathcal J},{\mathcal Q}_{\rm E}$ and ${\mathcal Q}_{\rm M}$. Therefore we have
\begin{equation}
A^{2}(S)\geq A^{2}(S_{1})\geq 16\pi^2\big[4{\mathcal J}^2+({\mathcal Q}_{\rm E}^2+ {\mathcal Q}_{\rm M}^2)^2\big]=16\pi^2\big[4J^2+(Q_{\rm E}^2+ Q_{\rm M}^2)^2\big]           
\end{equation}  
as desired.
\end{proof}

\subsection{A discussion on the extreme Kerr-Newman sphere}\label{seckerr}

We have seen above, that the set $\mathcal D_0$ given by equations \eqref{EKN1}-\eqref{EKN4} plays a crucial role in bounding the area of an axisymmetric MOTS or minimal surface (over a maximal slice), and moreover, due to Proposition \ref{trumpet}, in bounding the area of any surface in axially symmetric 
electrovacuum initial data. We show here how this set is related to extreme Kerr-Newman solution, from where it takes the name \textit{extreme Kerr-Newman sphere set}.

It is well known that the Kerr-Newman solution is parametrized by four quantities: the mass $m$, the angular momentum $J$ and the electromagnetic charges $Q_{\rm E}$, $Q_{\rm M}$. Of these parameters, let $J, Q_{\rm E}, Q_{\rm M}$ be fixed and decrease the remaining parameter $m$ as
$m\downarrow m_{0}$.
If we denote by ${\mathcal{D}}_{m}$ the set on the bifurcating sphere (for each $m$) then the limit $\lim_{m\downarrow m_{0}} {\mathcal{D}}_{m}={\mathcal{D}}_{0}$ is obtained. In other words we take the limit of ${\mathcal{D}}_{m}$ as the black holes become extremal to obtain ${\mathcal{D}}_{0}$ in (\ref{EKN1})-(\ref{EKN4}). As we discuss below, this way of finding ${\mathcal{D}}_{0}$  allows one to see how this particular kind of datum arises in the contexts {\bf \small O1} and {\small\bf O2} mentioned in section \ref{main}. 

The spacetime metrics for the Kerr-Newman solutions, in the usual Boyer-Lindquist coordinates, are given by (see \cite{Carter72})
\begin{align}\label{4metricKN}
g_{ab}dx^a dx^b=&-\frac{\Delta-a^2\sin^2\theta}{\Sigma}dt^2-\frac{2a\sin^2\theta}{\Sigma}(r^2+a^2-\Delta)dtd\phi \nn\\
&+\frac{(r^2+a^2)^2-\Delta a^2\sin^2\theta}{\Sigma}\sin^2\theta d\phi^2
+\frac{\Sigma}{\Delta}dr^2+\Sigma d\theta^2, 
\end{align}
where
\begin{equation}
\Sigma:=r^2+a^2\cos^2\theta,\qquad \Delta:=r^2+a^2+Q^2-2mr.
\end{equation}
The parameter $a=J/m$ is the angular momentum per unit mass and again $Q^2=Q_{\mathrm E}^2+Q_{\mathrm M}^2$. The electromagnetic part of the solution is encoded in the potential $A_a$ which is given explicitly by (see \cite{Carter72})
\begin{equation}\label{A}
A_a=-\frac{Q_{\mathrm E}r}{\Sigma}[(dt)_a-a\sin^2\theta(d\phi)_a]+\frac{Q_{\mathrm M}\cos\theta}{\Sigma}[a(dt)_a-(r^2+a^2)(d\phi)_a]
\end{equation}
The {\it subextremal} Kerr-Newman black holes are those solutions with $m^{2} > \frac{Q^{2}+\sqrt{4J^{2}+Q^{4}}}{2}$.
The {\it extreme Kerr-Newman black holes} are those solutions with $m^{2} = \frac{Q^{2}+\sqrt{4J^{2}+Q^{4}}}{2}$.
Let us concentrate on non-extreme Kerr-Newman black holes. Let $r_{H}$ be the greatest root of $\Delta=0$ (corresponding
to the event horizon), explicitly $r_{H}=m+\sqrt{m^2-a^2-Q^2}$. The range of coordinates $\{r\geq r_{H}\}$ ($t\in \mathbb{R},\theta\in [0,\pi ),\varphi\in [0,2\pi)$ arbitrary) covers exactly the whole domain of outer-communication and its boundary $\{r=r_{H}\}$ consists of a bifurcating sphere, a black hole horizon and a white hole horizon (respectively BHH and WHH, see Figure \ref{EKerr}).
The bifurcation surface is located at $\{r=r_{H}\}$ over the maximal slice $\{t=0\}$. It has a double character: it is at the same time a strictly stable minimal surface over the (doubling of the) maximal slice $\{t=0\}$ and a strictly stable MOTS on the space-time.  The area of the bifurcating sphere is easily calculated (use that $\Delta(r_{H})=0$) as
\be
A(S_{H})= 4\pi (r_{H}^{2}+a^{2})>4\pi \sqrt{4J^{2}+Q^{4}}  ,
\ee
and we have
\be\label{2}
A(S_{H})\downarrow 4\pi \sqrt{4J^{2}+Q^{4}}  ,
\ee
as $m\downarrow m_{0}$. 
\begin{figure}[h]
\centering\includegraphics[width=10cm,height=7.5cm]{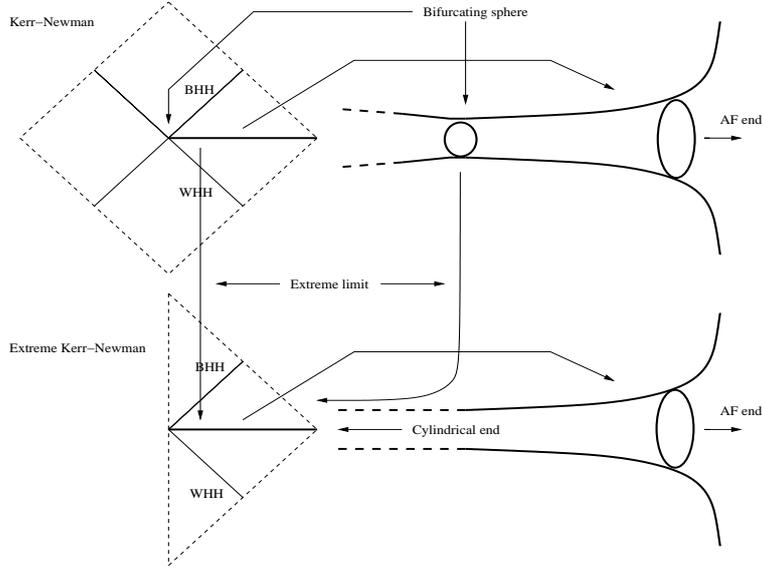}
\caption{Penrose diagram of the Kerr-Newman solution and its $\{t=0\}$ slice. Below, the Penrose diagram of the extreme Kerr-Newman solution and its $\{t=0\}$ slice displaying the cylindrical end and the asymptotically flat (AF) end.}
\label{EKerr}\end{figure}

We make now some claims, crucial to link the MOTS and minimal surface perspectives and show why the limit ${\mathcal{D}}_{m}\rightarrow {\mathcal{D}}_{0}$ allows us to see the set ${\mathcal{D}}_{0}$ as in {\bf\small O1} and {\bf\small O2}.

\begin{enumerate}
\item[I.] The set ${\mathcal{D}}_{m}$ over $S_{H}$ is the same as the  set on any axisymmetric sphere $S$  embedded 
in the black (white) hole horizon.  This can be seen as follows. The past (future) space-time flow generated by the stationary Killing field pushes any surface $S$ over the black (white) hole towards $S_{H}$ and the convergence is smooth. As the flow by the Killing vector field is an isometry (also leaving $F_{ab}$ invariant) and the components $(\sigma,\omega,\psi,\chi)$ of the set on the surface $S$ are intrinsic to the surface, it follows by this and continuity that the set over $S_{H}$ or over any axisymmetric sphere $S$ must be the same.

\item[II.] The black (white) hole horizon of the extreme Kerr-Newman solution is the limit of the black (white) hole horizon of the Kerr-Newman black hole solutions as $m\downarrow m_{0}$. To see this, just take the point-wise limit of expression (\ref{4metricKN}). In this limit the horizons $\{r=r_{H}\}$ approach the extreme horizon $\{r=m_{0}=\sqrt{Q^{2}+a_{0}^{2}}\}$. 

\item[III.] For every $m$, consider the initial data over $\Sigma=\{t=0\}$, $(\Sigma, h_{m},K_{m};E_{m},B_{m})$ where we put a subindex $m$ to emphasize that the initial datum is parametrized by $m$. ``Following" the initial data around $S_{H}$ as $m\downarrow m_{0}$, a smooth limit initial datum is obtained (see below for details on how to perform the limit). It is the so called {\it extreme Kerr-Newman throat}, which is a maximal electrovacuum initial datum on $\mathbb{R}\times S^{2}$ with the explicit form
\begin{align}
\label{K1} h_{T}&=\Sigma_{0} d\tilde{r}^{2}+\Sigma_{0} d\theta^{2}+\frac{(4J^{2}+Q^{4})}{\Sigma_{0}}\sin^{2}\theta d\varphi^{2},\\
\label{K2} K_{T}&= -J((J/a_{0})^{2}+a_{0}^{2})\frac{\sin^{2}\theta}{\Sigma_{0}^{\frac{3}{2}}}(d\tilde{r}dt+dtd\tilde{r}),\\
\label{K3} E_{T}&=-\bigg[ Q_{\rm E}(Q^{2}+a_{0}^2\sin^2\theta)-Q_{\mathrm M}2J\cos\theta\bigg]  
\frac{ d{\tilde{r}}}{\Sigma_{0}^{\frac{3}{2}}},\\
\label{K4} B_{T}&=-\bigg[Q_{\mathrm M}(Q^2+a_{0}^2\sin^2\theta)+Q_{\mathrm E}2J\cos\theta \bigg] 
\frac{ d{\tilde{r}}}{\Sigma_{0}^{\frac{3}{2}}}
\end{align}
The solution is independent 
of $\tilde{r}$ (the coordinate in the $\mathbb{R}$ factor; see below), which implies that $\partial_{\tilde{r}}$ is a Killing field. For this reason the initial datum has the same form if we replace $\tilde{r}$ by $\tilde{r}+c$ where $c$ is a constant. In particular the coordinate can be chosen in such a way that the bifurcating sphere $S_{H}$ (for given $m$) converges (as $m\downarrow m_{0}$) to the minimal sphere $S_{0}=\{\tilde{r}=0\}$, that we define as an {\it extreme Kerr-Newman throat sphere} and which because of (\ref{2}) satisfies (\ref{ine}). Of course any other sphere with constant $\tilde{r}$ has the same set of potentials ${\mathcal{D}}_{0}$.

We emphasize that the calculations leading to the extreme Kerr-Newman throat initial datum (\ref{K1})-(\ref{K4}) are long but straightforward if one follows a simple procedure. From (\ref{4metricKN})-(\ref{A}) obtain the explicit expressions of $h$, $K$, $E$ and $B$, over $\{t=0\}$ in the coordinates $\{r,\theta,\varphi\}$ ($r\geq r_{H}$). Then make the change of the radial coordinate $r$ to $\tilde{r}$ as
\be
\tilde{r}(r)=\int_{r_{H}}^{r}\frac{1}{\sqrt{\Delta(\bar{r})}} d\bar{r}
\ee
Of course $\tilde{r}(r_{H})=0$. Express $h,K,E,B$ whose components where given in terms of $\{(r,\theta,\varphi)\}$, in the coordinates $\{(\tilde{r},\theta,\varphi)\}$. Note that now the range of the coordinates $\{(\tilde{r},\theta,\varphi)\}$ is $[0,\infty)\times [0,\pi)\times [0,2\pi)$. Then {\it in this domain} take the point-wise limit $m\downarrow m_{0}$ of everyone of the components of the fields (in the $\{(\tilde{r},\theta,\varphi)\}$ coordinates). The result is (\ref{K1})-(\ref{K4}).  

\end{enumerate}

\noindent Summarizing, from I, II and III one obtains that the set ${\mathcal{D}}_{0}=\lim_{m\downarrow m_{0}} {\mathcal{D}}_{m}$ verifying (\ref{ine}), can be achieved  as the set on a MOTS inside a space-time (more precisely on an axisymmetric surface over the horizon of the extreme Kerr-Newman solution), 
or as the set endowed on stable minimal surfaces over maximal slices (more precisely over the extreme Kerr throat initial datum).

To see that ${\mathcal{D}}_{0}$ is given by (\ref{EKN1})-(\ref{EKN4}) proceed as follows. Making $\tilde{r}=0$ in (\ref{K1}) one obtains the two-metric of the extreme Kerr-Newman sphere to be 
\begin{equation}
ds^2=\Sigma_{0}d\theta^{2}+\frac{(4J^{2}+Q^{4})}{\Sigma_{0}}\sin^{2}\theta d\varphi^{2}
\end{equation}
From the definition of $\sigma$ in (\ref{SDE}) one obtains (\ref{EKN1}). To obtain (\ref{EKN3}) and (\ref{EKN4}) use (\ref{K3}) and (\ref{K4}) and the definitions (\ref{defpsichi}). We discuss now how to obtain (\ref{EKN2}). Over any two-sphere $\{r=r_{1}>r_{H},t=0\}$ on a Kerr-Newman black-hole, one uses the null vectors 
\bea\label{ell}
\ell^a&=&\left(r^2+a^2\right)(\partial_t)^a+a(\partial_\phi)^a+\Delta(\partial_r)^a \\
\label{k}
k^a&=&\left(\frac{r^2+a^2}{2\Delta\Sigma}\right)(\partial_t)^a+\frac{a}{2\Delta\Sigma}(\partial_\phi)^a-\frac{1}{2\Sigma}(\partial_r)^a \ ,
\eea
normalized such that $\ell^ak_a=-1$ to calculate $\Omega^{(\ell)}_a$. Taking the limit $r_{1}\rightarrow r_{H}$ and then the limit $m\downarrow m_{0}$ one obtains a limit form over the extreme Kerr-throat, which can be calculated to be
\bea
\label{e:Omega_extreme}
\Omega^{(\ell)} = -\frac{1}{(2\Sigma_0)^2}\left(2 a_{0}^2\Sigma_0 \sin(2\theta)d\theta + 4a_{0}\sqrt{a_{0}^2+Q^2}(2a_{0}^2+Q^2) 
\sin^2(\theta) d\phi\right)
\eea
 From (\ref{e:Omega_ell}) and axial symmetry\footnote{ More generally, one can fix 
$\tilde{\omega}$ and $\lambda$ by solving the second-order system: 
$D^aD_a \tilde\omega = - f$, $ D^aD_a \lambda = D^a \Omega^{(\ell)}_a$, 
where $(d\Omega^{(\ell)})_{ab} = f \epsilon_{ab}$.}, by
solving $\partial_{\theta}\tilde{\omega}=\Omega^{(\ell)}_{\phi}$ and 
$\partial_{\theta} \lambda=\Omega^{(\ell)}_{\theta}$, and taking into account 
$\omega_{K0}'=2\eta \tilde\omega^{'}_{0}$, we get
\bea
\label{e:omega_lambda}
 \omega_{K0}(\theta)&=&
\frac{(2a_0^2+Q^2)^2Q^2}{a_0^2(a_0^2+Q^2)}\arctan\left(\frac{a_0\cos\theta}{\sqrt{a_0^2+Q^2}}\right)-
\cos\theta\frac{(2a_0^2 + Q^2)^3}{\sqrt{a_0^2+Q^2}a_0\Sigma_0}  \nn \\
\lambda(\theta)&=&\ln[\sqrt{2\Sigma_0}] \ .
\eea
Moreover, we verify $\omega_{K0}(\theta=0)=-\omega_{K0}(\theta=\pi)=-4J_K$, where $J_K$ is the Komar contribution to the total angular momentum 

\begin{eqnarray}
J_K=\frac{(2a_0^2 + Q^2)^2}{4a_0^2(a_0^2 + Q^2)}\left[a_0\sqrt{a_0^2+Q^2}-Q^2\arctan\left(\frac{a_0}{\sqrt{a_0^2+Q^2}}\right)\right]
 \ .
\end{eqnarray}
Using expression (\ref{e:omega_lambda}), together with (\ref{EKN3}) and (\ref{EKN4})
into (\ref{defomega}), we get $\omega_0$ in (\ref{EKN2}), and thus complete the derivation of the set $\mathcal D_0$.

\vs

Finally, we present two remarks concerning the extreme Kerr- Newman sphere:
\begin{itemize}
\item \textit{$\mathcal D_0$ in $\mathbb H^2_{\mathbb C}$}. There is an interesting description of the geometry of the extreme Kerr-Newman sphere which shows the underlying connection with the complex hyperbolic space. This connection will be exposed in section \ref{prueba2} and arises when one studies the critical point of the functional $\mathcal M$. What we want to show here is that the set $\mathcal D_0$ can be visualized as two arcs of circles in $\mathbb H^2_{\mathbb C}$.  In order to describe these arcs, we consider, instead of the quadruple $(\sigma,\omega,\psi,\chi)$, the two pairs $(\zeta,\omega)$ and $(\psi,\chi)$ where
\begin{equation}
\zeta=-(\eta+\psi^{2}+\chi^{2}).
\end{equation} 
Then, whether by working with the Euler-Lagrange equations of $\mathcal M$  (as is done in section \ref{prueba2}, precisely, the form of the r.h.s of (\ref{LTA}) comes from the last two eqs. of (\ref{EQchi}), while the form of l.h.s comes from the first eqn. of (\ref{EQchi}) and (\ref{x4})) or with the explicit expression for the potentials, equations \eqref{EKN1}-\eqref{EKN4}, we find the following remarkable relations:
\begin{equation}\label{LTA}
\zeta+i\omega=R_1e^ {if}+ B_1, \qquad \chi+i\psi=R_2e^ {if}+ B_2,
\end{equation}
where the angle to the center $f$ is given by (see \eqref{fff})
\be\label{anglef}
f=2\arctan \left(\frac{\sqrt{Q^4+4J^2}-Q^2}{2J} \cos\theta\right)
\ee
and $R_1=-2\sqrt{4J^{2}+Q^{4}}$, $B_1=Q^{2}$, $R_2=-\frac{Q}{2J}\sqrt{4J^{2}+Q^{4}}$, $B_2=Q^3/2J$.

This shows that the first arc, in the $(\zeta,\omega)$ plane, starts at $\left(-Q^{2},-4|J|\right)$  and ends at $\left(-Q^{2},4|J|\right)$ (this can be obtained by evaluating the pair at the values $0, \pi$ respectively). The center of the circle to 
which the arc belongs lies on the $\zeta=0$ axis and its radius is $R_1$. The arc in the $(\psi,\chi)$ plane starts at $(-Q,0)$ and ends at $(Q,0)$. The center of the circle to which the arc belongs lies on the $\psi=0$ axis and its radius is $R_2$. 

\begin{figure}[h]
\centering\includegraphics[width=10cm,height=5cm]{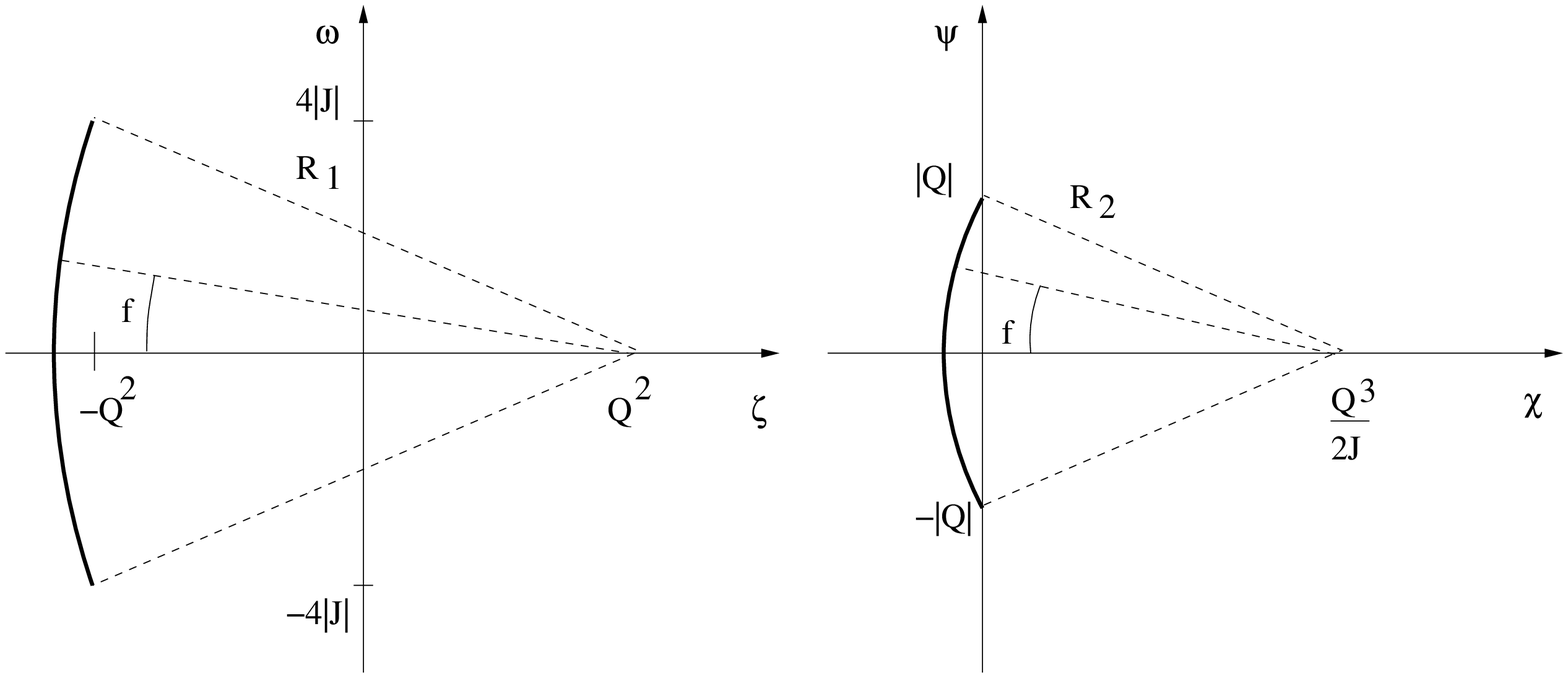}
\caption{$\mathcal D_0$ as arcs of circles.}
\label{Fig1}\end{figure}

\item \textit{On the choice of the probe function $\alpha$}. Here we want to give some insights about the choice of the function $\alpha$, equation (\ref{e:alpha}), entering in the stability condition \eqref{e:inequality_alpha_AJQ}. In particular, it has the nature of a rescaling factor between null normals and
we show that it is related to the minimizing set $\mathcal D_0$. From the transformation properties of $\Omega_a^{(\ell)}$ under  a rescaling of the null normals $\ell^a$, $k^a$, we note that the null vector $\ell_o^a=e^{-\lambda} \cdot \ell^a$, with $\lambda$ given by the expression in (\ref{e:omega_lambda}), is such that the associated fundamental form $\Omega_a^{(\ell_o)}$ is divergence-free, i.e. $D^a \Omega_a^{(\ell_o)}=0$. This provides a natural quasilocal normalization for the outgoing null vector on $S$. On the other hand, evaluating $\alpha$ in (\ref{e:alpha}) with the expressions in (\ref{EKN1}) we can check $\alpha \ell_o^a= \mathrm{const}\cdot\ell^a_{\mathrm{Killing}}$, where $\ell^a_{\mathrm{Killing}}=(\partial_t)^a-\Omega \partial_\phi$ is the only null vector on $S$ (up to constant) that extends as a Killing vector in a spacetime neighborhood of $S$ (here $\Omega$ is the constant horizon angular velocity). In other words, our choice of $\alpha$ in  (\ref{e:alpha}) provides precisely the rescaling from the canonical quasilocal choice $\ell^a_o$ on $S$, with  divergence-free fundamental form, to the globally defined Killing vector of the Kerr-Newman spacetime that becomes null on the horizon. This remark is explained by the rigidity results in \cite{Jaramillo:2012zi}
(see also the analysis in \cite{Mars:2012sb}).
\end{itemize}

\section{Different avenues to prove the AJQ inequality} \label{avenidas}

The AJQ inequality \eqref{ine} is obtained from two ingredients, namely, from the stability condition, leading to Lemma \ref{lema1}, and from the resolution of the naturally associated minimizing problem, leading to Lemma \ref{teo}. In this section, we show two different ways to approach the variational principle.

Before addressing these points, a remark on the implications of the analysis of the AJQ
inequality in the stationary case is in order. In Refs. \cite{HenAnsCed08,Hennig10,Ansorg:2010ru}  
the strict version of inequality (\ref{ine}), with vanishing magnetic charge,
is proved for Killing horizons in axisymmetric spacetimes. 
The scheme of that proof shares the two ingredients of the analysis in this section: first, 
use of a stability condition in the form of a horizon (outer) subextremal
assumption \cite{Hay94,Booth08} from which an {\em integral stability condition} for axisymmetric 
Killing horizons is derived; second, definition of a variational 
problem from such integral stability condition, whose resolution leads to the strict (\ref{ine}).
Remarkably, in Ref. \cite{GabachJaramillo11} it is explicitly shown that the first ingredient, namely
the integral stability condition, can be derived directly from the  quasilocal
(strict) stability of axisymmetric MOTS, in particular from the strict version of inequality 
(\ref{e:inequality_alpha}). Further geometrical insight on the relation between 
the stationary axisymmetric black hole condition and quasilocal MOTS stability is provided in 
\cite{Chrusciel11,Mars:2012sb}. As a consequence, the variational analysis in 
\cite{HenAnsCed08,Hennig10,Ansorg:2010ru} can be exactly applied to strictly stable axisymmetric
MOTS, so that the proof in \cite{HenAnsCed08,Hennig10,Ansorg:2010ru} of the strict inequality
(\ref{ine}) with $Q_{\rm M}=0$ extends straightforwardly from the stationary setting 
to the dynamical case with arbitrary standard matter \cite{GabachJaramillo11} (namely,
the strict version of item I in Th. \ref{coromots}). The extension of the variational problem
in \cite{HenAnsCed08,Hennig10,Ansorg:2010ru} to include the equality case and the 
rigidity analysis is under research \cite{Ceder12}.

Following a different rationale, the two approaches to the variational problem discussed in this section 
aim at enriching the understanding of the geometric structure underlying 
Theorem \ref{coromots}. We believe that each of the perspectives presented here, 
gives important insights about this problem.

The first approach, in section \ref{prueba1} deals with non-necessarily axisymmetric sets of potentials $\mathcal D$ and its associated functional $\mathcal M$. Although the main result, Theorem \ref{coromots} holds in the physical scenario of axisymmetric surfaces $S$, the fact that the variational problem can be stated and solved outside axisymmetry shows that extreme Kerr-Newman sphere plays a special role among a wider class of sets $\mathcal D$. Inspired by this generalization, one is tempted to think about the possibility of extending inequality \eqref{ine} to other  non-necessarily axisymmetric physical situations. This, however is not an easy task, mainly because it is not clear how to give a satisfactory  canonical definition of angular momentum outside axial symmetry.  Nevertheless, if such statement can be made, the functional $\mathcal M$ and its properties studied here might be of relevance.

The second approach, in section \ref{prueba2}, is restricted to axisymmetry and therefore, when solving the minimization problem for $\mathcal M$, the Euler-Lagrange equations reduce to a system of ordinary differential equations which can be solved  explicitly. Then, a remarkable point that comes up when studying these equations, is that the boundary conditions $J$, $Q_{\rm E}$ and $Q_{\rm M}$ for the minimizer of $\mathcal M$ determine uniquely the boundary conditions for the remaining potential $\sigma$. This is the key fact under the sharpness of inequality \eqref{ine}. Actually, an important consequence of this is that we can prove uniqueness for the minimizer of $\mathcal M$ with given values of $J, Q_{\rm E}, Q_{\rm M}$ without any reference to the boundary values of $\sigma$. This is a difference to what we do in the non-axisymmetric case, where the boundary values of $\sigma$ are prescribed from $A=4\pi e^{\sigma}|_{\theta=0}$.

\subsection{Proof from harmonic maps}\label{prueba1}

In this section we prove Lemmas \ref{teo} and \ref{unic} by exploiting the connection between $\mathcal M$ and a harmonic energy for maps from the sphere into the complex hyperbolic space. The first Lemma follows closely the arguments  given by Acena et al, \cite{Acena11}. To prove the rigidity in inequality \eqref{ine} we use certain properties of the distance between harmonic maps in the complex hyperbolic space.

\begin{proof} (\textit{of Lemma \ref{teo}})

To prove our claim, we follow the lines and  arguments of \cite{Acena11}, and refer to that article for more details. The key points in the argument are the following:
\begin{enumerate}
\item The extreme Kerr-Newman sphere, \textit{i.e.} the set $\mathcal D_0$, satisfies the Euler-Lagrange equations for the functional $\mathcal M$:
\begin{equation}\label{ech1}
\Delta\sigma-2=-\frac{(D\omega +2\chi D\psi-2\psi D\chi)^2}{\eta^2}-\frac{2}{\eta}\left((D\psi)^2+(D\chi^2)\right)^2,
\end{equation}
\begin{equation}\label{ecuomega}
D_a\left(\frac{D^a\omega+2\chi D^a\psi-2\psi D^a\chi}{\eta^2}\right)=0
\end{equation}
\begin{equation}
D_a\left(\frac{D^a\chi}{\eta}\right)-\frac{1}{\eta^2}D_a\psi(D^a\omega+2\chi D^a\psi-2\psi D^a\chi)=0 
\end{equation}
\begin{equation}\label{echf}
D_a\left(\frac{D^a\psi}{\eta}\right)+\frac{1}{\eta^2}D_a\chi(D^a\omega+2\chi D^a\psi-2\psi D^a\chi)=0,
\end{equation}
where indices are moved with the standard round metric on $S^2$.

\item The functional $\mathcal M$ is related to the harmonic energy $\tilde{\mathcal M}_\Omega$ for maps $(\eta,\omega,\chi,\psi)$ from a subset $\Omega\subset S^2\setminus\{\theta=0,\pi\}$ into the complex hyperbolic space $\mathbb {H} ^2_{\mathbb C}$ which is equipped with the metric 
\begin{equation}\label{hyperbolicmetric}
g_H=\frac{d\eta^2}{\eta^2}+\frac{(d\omega+2\chi d\psi-2\psi d\chi)^ 2}{\eta^2}+4\frac{d\chi^2+d\psi^2}{\eta},
\end{equation}
and is given by
\begin{equation}
\tilde{\mathcal M}_\Omega=\frac{1}{2\pi}\int_\Omega \frac{|D\eta|^2}{\eta^2}+\frac{|D\omega+2\chi D\psi-2\psi D\chi|^ 2}{\eta^2}+4\frac{|D\chi|^2+|D\psi|^2}{\eta} dS_0.
\end{equation}
Now restrict the integral in the definition of $\mathcal M$, \eqref{defM} to compact regions with smooth boundary $\Omega\subset S^2\setminus\{\theta=0,\pi\}$ and denote the resulting functional as $\mathcal M_\Omega$. We have
\begin{equation}\label{relMMt}
\tilde{\mathcal M}_\Omega=\mathcal M_\Omega+4\int_\Omega\ln\sin\theta dS+\oint_{\partial\Omega}(4\sigma+\ln\sin\theta)\frac{\partial\ln\sin\theta}{\partial n}dl
\end{equation} 
where $n$ is the exterior unit normal to the boundary $\partial\Omega$ of $\Omega$ and $dl$ is the measure element on $\partial\Omega$. Since the difference between $\tilde{\mathcal M}_\Omega$ and $\mathcal M_\Omega$ is a constant plus a boundary term, both functionals have the same Euler-Lagrange equations.

\item A result of Hildebrandt \textit{et al} \cite{Hildebrandt77} states that if the domain for the map is compact, connected, with non-void 
boundary and the target manifold has negative sectional curvature, then a minimizer of the harmonic energy with Dirichlet boundary conditions
 exists, is unique, smooth and satisfies the associated Euler-Lagrange equations. That is, harmonic maps are minimizers of the harmonic energy
 for given Dirichlet boundary conditions.
\end{enumerate}

With the above comments, the proof goes as follows: divide the sphere into three regions as indicated in equations \eqref{div1}.  Use a partition function to interpolate the potentials between extreme Kerr-Newman solution in region $\Omega_I$ and a general solution in region $\Omega_{III}$. This gives a Dirichlet problem in region $\Omega_{IV}=\Omega_{II}\cup\Omega_{III}$, which implies, by point 3. above, that the mass functional for extreme Kerr-Newman is less than or equal to the mass functional for the auxiliary interpolating map in the whole sphere. Finally, we take the limit as $\Omega_{III}$ covers the whole sphere and show that the mass functional for the auxiliary maps converges to the mass functional for the original general set.

After giving this general discussion about the proof, we begin with the splitting of the sphere according to
\begin{equation}\label{div1}
\Omega_I = \{\sin\theta\leq e^{-(\log\epsilon)^2}\},\quad \Omega_{II} = \{e^{-(\log\epsilon)^2}\leq\sin\theta\leq\epsilon\},\quad
\Omega_{III} = \{\epsilon\leq\sin\theta\},
\end{equation}
where $0<\epsilon<1$. We also define the region $\Omega_{IV} = \Omega_{II}\cup\Omega_{III}$.

Let $f:\mathbb{R}\rightarrow\mathbb{R}\in C^\infty(\mathbb{R})$, $0\leq f\leq 1$, be
the partition function defined as
\begin{equation}
f(t)=1\quad \mbox{for}\quad t\leq1,\qquad f(t)=0\quad \mbox{for}\quad 2\leq t,\qquad\left|\frac{d\,f}{dt}\right|\leq1
\end{equation}
and $f_\epsilon$ be
\begin{equation}
 f_\epsilon(\rho)=f(t_\epsilon(\rho)),\qquad  t_\epsilon(\rho)=\frac{\log(-\log\rho)}{\log(-\log\epsilon)},\qquad \rho\leq1.
\end{equation}
Therefore we have

\begin{equation}
f_\epsilon(\rho)=0\quad \mbox{for}\quad\rho\leq e^{-(\log\epsilon)^2},\qquad f_\epsilon(\rho)=1\quad\mbox{for}\quad\rho\geq\epsilon
\end{equation}
and
\begin{equation}\label{intChi}
 \lim_{\epsilon\rightarrow0}\int_0^\infty|\partial_\rho f_\epsilon|^2\rho d\rho=0.
\end{equation}

Now we define the interpolating functions. Let $u$ represent any of the variables $\sigma$, $\omega$, $\chi$, $\psi$, and let $u_0$ represent any of the variables $\sigma_0$, $\omega_0$, $\chi_0$, $\psi_0$, corresponding to the extreme Kerr-Newman sphere set with the same angular momentum and charges. We define $u_\epsilon$ to be
\begin{equation}
\label{uepsilon}
u_\epsilon=f_\epsilon(\sin\theta)\,u+(1-f_\epsilon(\sin\theta))\,u_0=(u-u_0)f_\epsilon(\sin\theta)+u_0.
\end{equation}
This gives $u_\epsilon|_{\Omega_I}=u_0|_{\Omega_I}$ and $u_{\epsilon}|_{\Omega_{III}}=u|_{\Omega_{III}}$. We also define 
\begin{equation}
\label{Mepsilondef}
{\cal
  M}^\epsilon= \mathcal M (\sigma_\epsilon,\omega_\epsilon,\psi_\epsilon,\chi_\epsilon),   
\end{equation}
and correspondingly ${\cal M}^\epsilon_\Omega$ and $\tilde{\cal
  M}^\epsilon_\Omega$ when the domain of integration is restricted to some region $\Omega$. We also denote by a superscript `$0$' these quantities
calculated for $u_0$.

We have all the ingredients needed to make use of the result in \cite{Hildebrandt77}. For this, let us consider now a fixed value of $\epsilon$, and the functions
$(\sigma,\omega,\chi,\psi)$ on the set $\Omega_{IV}$. By \cite{Hildebrandt77} we know that there exists one and only one set of functions that minimizes $\tilde{\cal M}$ on $\Omega_{IV}$ for
given boundary data, and that this function satisfies the Euler-Lagrange equations of $\tilde{\cal M}$ on $\Omega_{IV}$. By construction of $u_\epsilon$ we have that $u_\epsilon$ and $u_0$ have the same boundary values on $\Omega_{IV}$,
\begin{equation}
u_\epsilon|_{\partial\Omega_{IV}}=u_0|_{\partial\Omega_{IV}}.
\end{equation}
As we already know that $u_0$ is a solution of the Euler-Lagrange equations
of $\mathcal M$, and thus of $\tilde{\cal M}$ there, then $u_0$ is the only minimizer of $\tilde{\cal M}$ on $\Omega_{IV}$ with these boundary conditions. This means that $\tilde{\cal M}^\epsilon_{\Omega_{IV}}\geq\tilde{\cal M}^0_{\Omega_{IV}}$. Both ${\cal M}$ and $\tilde{\cal M}$ are well defined on $\Omega_{IV}$, and by \eqref{relMMt} their difference is just a constant. Therefore we also have ${\cal M}^\epsilon_{\Omega_{IV}}\geq{\cal M}^0_{\Omega_{IV}}$.

As we have already noted, $u_\epsilon|_{\Omega_{I}}=u_0|_{\Omega_{I}}$, and therefore ${\cal M}^\epsilon_{\Omega_I}={\cal M}^0_{\Omega_I}$. This together with the inequality in $\Omega_{IV}$ found above and the fact that $S^2=\Omega_I\cup\Omega_{IV}$ give
\begin{equation}\label{ineqMe}
{\cal M}^\epsilon\geq{\cal M}^0.
\end{equation}
Only the last step of the proof is lacking, that is, to show 
\begin{equation}\label{limitMe}
\lim_{\epsilon\rightarrow 0}{\cal M}^\epsilon={\cal M}.
\end{equation}
We write
\begin{equation}\label{m2}
  {\mathcal M}^\epsilon =\mathcal M_{\Omega_I}^\epsilon +\mathcal M_{\Omega_{II}}^\epsilon+\mathcal M_{\Omega_{III}}^\epsilon=\mathcal M_{\Omega_I}^0 +\mathcal M_{\Omega_{II}}^\epsilon+\mathcal M_{\Omega_{III}}
\end{equation}
Using the Dominated Convergence Theorem  it is not hard to see that the first integral in \eqref{m2} vanishes in the limit $\epsilon\to0$, since the domain reduces to the poles and we know that $\mathcal M_0$ is finite. Also, the third term in \eqref{m2} tends to $\mathcal M$ as $\Omega_{III}$ extends to cover the whole sphere.

To show that the second term in \eqref{m2} vanishes in the limit $\epsilon\to0$ we consider its different parts separately. We have
\begin{equation}
\label{Mepsilon}
{\cal
  M}^\epsilon_{\Omega_{II}}=\frac{1}{2\pi}\int_{\Omega_{II}}\left[|D\sigma_\epsilon|^2+4\sigma_\epsilon+\frac{|D\omega_\epsilon+2\chi_\epsilon D\psi_\epsilon-2\psi_\epsilon D\chi_\epsilon|^2}{e^{2\sigma_\epsilon}\sin^4\theta}+4\frac{|D\psi_\epsilon|^2+|D\chi_\epsilon|^2}{e^{\sigma_\epsilon}\sin^2\theta}\right]dS_0,   
\end{equation}
Using the definition of $u_\epsilon$ \eqref{uepsilon} we compute
\begin{equation}
Du_\epsilon= (u-u_0)Df_\epsilon+(Du-Du_0)f_\epsilon+Du_0.
\end{equation}
We see that
\begin{equation}\label{estimacion1}
\sigma_\epsilon\leq C
\end{equation}
because $\sigma$ and $\sigma_0$ are finite on $S^2$, and $f_\epsilon\leq1$. Here, and in what follows, we denote by $C$, $C_i$  constants independent of $\epsilon$. Also, because of the regularity of $\sigma$ and $\sigma_0$, we have 
\begin{equation}\label{estimacion2}
|D\sigma_\epsilon|^ 2\leq 3|Df_\epsilon|^ 2(\sigma-\sigma_0)^ 2+3|D\sigma-D\sigma_0|^ 2+3|D\sigma_0|\leq C_1|Df_\epsilon|^2+C_2.
\end{equation} 
Then, from \eqref{estimacion1}-\eqref{estimacion2} and using strongly the property \eqref{intChi} to bound the integral of $|Df_\epsilon|^2$, we conclude that the first two terms in $\mathcal M^\epsilon_{\Omega_II}$ go to zero as $\epsilon\to0$.

Now we work with the term
\begin{equation}
\int_{\Omega_{II}}\frac{|D\omega_\epsilon+2\chi_\epsilon D\psi_\epsilon-2\psi_\epsilon D\chi_\epsilon|^2}{e^{2\sigma_\epsilon}\sin^4\theta}dS_0.
\end{equation}
Using the fact that $f_\epsilon$ is bounded, and $\sigma, \sigma_0$ are regular we have 
\begin{eqnarray}
\frac{|D\omega_\epsilon+2\chi_\epsilon D\psi_\epsilon-2\psi_\epsilon D\chi_\epsilon|^2}{e^{2\sigma_\epsilon}\sin^4\theta}&\leq& C_1\frac{|Df|^2\left(\omega-\omega_0+2\psi_0\chi-2\chi_0\psi\right)^2}{\sin^4\theta}\nonumber\\
&+&C_2\frac{|D\omega+2\chi_0 D\psi-2\psi_0 D\chi|^2}{\sin^4\theta}\nonumber\\
&+&C_3\frac{|D\omega_0+2\chi_0 D\psi_0-2\psi_0 D\chi_0|^2}{\sin^4\theta}
+C_4\frac{|D\psi_0|^ 2(\chi-\chi_0)^ 2}{\sin^4\theta}\nonumber\\&+&C_5\frac{|D\chi_0|^ 2(\psi-\psi_0)^2}{\sin^4\theta}
+C_6\frac{(\chi-\chi_0)^2|D\psi-D\psi_0|^2}{\sin^4\theta}\nonumber\\&+&C_7\frac{|D\chi-D\chi_0|^2(\psi-\psi_0)^2}{\sin^4\theta}\label{terminoomega}
\end{eqnarray}

The term accompanying the constant $C_3$ is also pointwise bounded in $\Omega_{II}$ because extreme Kerr-Newman sphere satisfies the regularity item ii) in Definition \ref{regularset}.

In virtue of Definition \ref{regularset}, we find that the remaining terms in \eqref{terminoomega} are uniformly bounded in $\Omega_{II}$. Altogether we derive
\begin{equation}
\int_{\Omega_{II}}\frac{|D\omega_\epsilon+2\chi_\epsilon D\psi_\epsilon-2\psi_\epsilon D\chi_\epsilon|^2}{e^{2\sigma_\epsilon}\sin^4\theta}dS_0\leq \int_{\Omega_{II}}C_1|Df_\epsilon|^2+C_2 dS_0.
\end{equation}
It is important to remark that the l.h.s. in the above inequality is bounded when the potentials $\omega, \psi, \chi$ are smooth functions on $S^2$ satisfying condition i) in Definition \ref{regularset} and are such that $\mathcal M$ is finite, that is, condition ii) is no longer necessary.

In the limit $\epsilon\to0$ this integral vanishes by property \eqref{intChi}. 

Finally, we look at the term 
\begin{equation}\label{ultimo}
\int_{\Omega_{II}}\frac{|D\chi_\epsilon|^2+|D\psi_\epsilon|^2}{e^{\sigma_\epsilon}\sin^2\theta}dS_0.
\end{equation}
We have, as in \eqref{estimacion2}
\begin{equation}\label{estimacion3}
\frac{|D\chi_\epsilon|^ 2}{e^{\sigma_\epsilon}\sin^2\theta}\leq C\frac{|Df|^ 2(\chi-\chi_0)^ 2+|D\chi-D\chi_0|^ 2+|D\chi_0|^2}{\sin^2\theta}\leq C_1|Df_\epsilon|^2+C_2.
\end{equation}
where, in the first inequality we have used the boundedness of $\sigma, \sigma_0$. A similar behavior is found for the second term in \eqref{ultimo}. Therefore, taking into account the property \eqref{intChi} the limit $\epsilon\to0$ of \eqref{ultimo} gives zero. 

We have shown
\begin{equation}
\lim_{\epsilon\to0}\mathcal M^\epsilon_{\Omega_{II}}=0
\end{equation}
and thus the limit \eqref{limitMe}. This, together with \eqref{ineqMe} completes the proof of the Lemma.
\end{proof}

Now we present the proof of Lemma \ref{unic}, stating the uniqueness of the minimizer for the area in inequality \eqref{ine}. This is done by exploiting the properties of the distance between harmonic maps in the complex hyperbolic space.

\begin{proof}\textit{(of Lemma \ref{unic})}

We follow the lines of Weinstein \cite{Weinstein90} and Dain \cite{Dainmass08}. By contradiction, assume that there exists another regular set $\mathcal D^1$ which saturates \eqref{ine}. Denote with a subscrit 1 the quantites referred to this set (and with a subscript 0 the quantities referred to $\mathcal D_0$). Then we have 
\begin{equation}\label{areas}
A_1=4\pi e^{\frac{\mathcal M_0-8}{8}}=A_0.
\end{equation} 
Then, using $A_1\geq 4\pi e^{\frac{\mathcal M_1-8}{8}}$  and \eqref{areas} we find $\mathcal M_1=\mathcal M_0$. This means that $\mathcal D_1$ is also a  critical point of the functional $\mathcal M$, i.e. it is a harmonic map.

By hypothesis, the second solution has the same values of the angular momentum and charges. Let $the exp\Gamma$ be the poles on $S^2$, so 
\begin{equation}\label{condcontorno}
\omega_1\big|_{\Gamma}=\omega_0\big|_{\Gamma}=\pm4J,\quad \chi_1\big|_{\Gamma}=\chi_0\big|_{\Gamma}=\pm Q_{\rm M},\quad \psi_1\big|_{\Gamma}=\psi_0\big|_{\Gamma}=\pm Q_{\rm E}.
\end{equation}
But also, in virtue of equation \eqref{areas}, we conclude that $\sigma_1|_{\Gamma}=\sigma_0|_{\Gamma}=\ln(A_0/4\pi)$ (recall that the area is determined solely by the value of $\sigma$ on the poles, through the expression $A=4\pi e^{\sigma(0)}$). 

In what follows we will prove that the distance between these two solutions is in fact zero, and thus that the solutions are identical.

Let ($\eta_1,\omega_1,\chi_1,\psi_1$) and  ($\eta_0,\omega_0,\chi_0,\psi_0$) be two harmonic maps $S^2\setminus\Gamma\to \mathbb H^2_\mathbb C$, and consider, for each ($\theta,\phi$), the corresponding points in $\mathbb H^2_{\mathbb C}$ equipped with the hyperbolic metric introduced above in equation \eqref{hyperbolicmetric}.
 The distance $d$ between these two points is  given by (see \cite{Weinstein90})
\begin{equation}\label{dist}
\cosh d=1+\delta
\end{equation}
where
\begin{eqnarray}\label{delta}
\delta=\frac{(\omega_0-\omega_1+2\chi_0\psi_1-2\chi_1\psi_0)^2+((\chi_0-\chi_1)^2+(\psi_0-\psi_1)^2))^2}{2\eta_1\eta_0}+\nonumber\\
+\left(\frac{1}{\eta_1}+\frac{1}{\eta_0}\right)[(\chi_0-\chi_1)^2+(\psi_0-\psi_1)^2]+\frac{(\eta_0-\eta_1)^2}{2\eta_1\eta_0}.
\end{eqnarray}
Therefore, since the functions $\omega, \chi, \psi,\sigma$ are regular on $S^2$,  $d$ defines a function $d:S^2\to\mathbb R$. We use the results of Shoen and Yau \cite{SchoenYau} to deduce that the square distance between harmonic maps is a subharmonic function on $S^2$, that is
\begin{equation}
 \Delta d^2\geq0,
\end{equation}
and, since $\delta$ is a convex function of $d^2$, then
\begin{equation}
 \Delta\delta\geq0.
\end{equation}

Let us see now that the distance between the two solutions at $\Gamma$ is zero. Begin with the first term in \eqref{delta}. From item ii) in Definition \ref{regularset}, the following behavior is deduced 
\begin{equation}\label{deriv2}
(\partial_\theta^2\omega+2\chi\partial_\theta^2\psi-2\psi\partial_\theta^2\chi)|_\Gamma=0
\end{equation}
(note that if the functions $\omega, \psi, \chi$ satisfy condition i) in Definition \ref{regularset} and $\mathcal M$ is finite, then the solutions of the Euler-Lagrange equations of $\mathcal M$ necessarily have the above behavior near the poles).

Then \eqref{deriv2} together with the boundary conditions give near the poles
\begin{equation}
\omega_0-\omega_1+2\chi_0\psi_1-2\chi_1\psi_0=O(\sin^3\theta)
\end{equation}
which implies
\begin{equation}\label{compomega}
\left.\frac{(\omega_0-\omega_1+2\chi_0\psi_1-2\chi_1\psi_0)^2}{2\eta_1\eta_0}\right|_\Gamma=0
\end{equation}

We look now at the second and third terms in \eqref{delta}. Since by hypothesis  $\partial_\theta\psi|_\Gamma=\partial_\theta\chi|_\Gamma=0$, we find
\begin{equation}
\psi_0-\psi_1=\chi_0-\chi_1=O(\sin^2\theta).
\end{equation}
Therefore we obtain
\begin{equation}\label{comppsichi}
\left.\left[\frac{[(\chi_0-\chi_1)^2+(\psi_0-\psi_1)^2)]^2}{2\eta_1\eta_0}+\left(\frac{1}{\eta_1}+\frac{1}{\eta_0}\right)[(\chi_0-\chi_1)^2+(\psi_0-\psi_1)^2]\right]\right|_\Gamma=0
\end{equation}
The last term we must investigate in \eqref{delta} is the one involving $(\eta_0-\eta_1)$. We write it as
\begin{equation}
\frac{(\eta_0-\eta_1)^2}{2\eta_0\eta_1}=\cosh(\sigma_0-\sigma_1)-1,
\end{equation}
but taking into account the boundary conditions $\sigma_1|_\Gamma=\sigma_0|_\Gamma$, we find
\begin{equation}\label{term}
\frac{(\eta_0-\eta_1)^2}{2\eta_0\eta_1}\big|_\Gamma=0.
\end{equation}
With conditions \eqref{compomega}, \eqref{comppsichi}, and \eqref{term} one verifies that 
\begin{equation}
\delta|_\Gamma=0.
\end{equation}

Then, since $\delta$ is continuous (and smooth on $S^2\setminus\{0\}$) and non negative, $\delta\big|_\Gamma=0$, and $\Delta\delta\geq0$ on $S^2\setminus\{0\}$ we can use the standard Maximum Principle to conclude that $\delta=0$ in $S^2$. Therefore $d=0$ and the two maps are identical. This completes the proof of the Lemma.
\end{proof}

\subsection{Proof from geodesics in $\mathbb H_{\mathbb C}^2$}\label{prueba2}

We prove now Lemmas \ref{teo} and \ref{unic} in the axially symmetric case, with zero magnetic charge, namely $Q_{\rm M}=0$. We make $Q= Q_{\rm E}$. The case when magnetic charge is present can be easily obtained by rotating along the $(\chi,\psi)$-plane, noting that rotations along the $(\chi,\psi)$-plane leave the functional $\mathcal M$, \eqref{defM}, invariant. We assume that either $J$ or $Q$ are non-zero otherwise there is nothing to prove.

The fundamental fact allowing to prove Lemmas  \ref{teo}-\ref{unic} only in terms of geodesics in the complex hyperbolic plane $\mathbb{H}^{2}_{\mathbb C}$ is the following identity  (see equation \eqref{relMMt})
\be\label{FUND}
\tilde{\mathcal M}_{t_{1},t_{2}}={\mathcal M}_{\theta_{1},\theta_{2}}+4\sigma \cos \theta\bigg|_{\theta_{1}}^{\theta_{2}}+4\ln \tan \frac{\theta}{2}\bigg|_{\theta_{1}}^{\theta_{2}}
\ee
where $t=\ln \tan \frac{\theta}{2}$ and
\begin{gather}\label{Energy}
\mathcal{M}_{\theta_{1},\theta_{2}}:=\int_{\theta_{1}}^{\theta_{2}}\bigg(\sigma'^{2}+4\sigma+\frac{(\omega'+2\chi\psi'-2\psi\chi')^{2}}{\eta^{2}}+4\frac{\psi'^{2}+\chi'^{2}}{\eta}\bigg)\sin\theta d\theta,\\
\tilde{\mathcal{M}}_{t_{1},t_{2}}:=\int_{t_{1}}^{t_{2}}\bigg( \frac{\dot\eta^{2}}{\eta^{2}}+\frac{(\dot \omega+2\chi\dot\psi-2\psi\dot\chi)^{2}}{\eta^{2}}+4\frac{\dot\chi^{2}+\dot\psi^{2}}{\eta}\bigg)\ dt
\end{gather}
 Equation \eqref{FUND}  shows that for fixed Dirichlet boundary conditions, critical points of ${\mathcal M}_{\theta_{1},\theta_{2}}$ are critical points of $\tilde{\mathcal M}_{t_{1},t_{2}}$ and vice versa. Now, consider $\gamma(t)=(\eta,\omega,\psi,\chi)(t)$, in $\mathbb{H}_{\mathbb{C}}^{2}$ with metric $g_H$ given by \eqref{hyperbolicmetric}. Then we have the remarkable relation
\be
\tilde{\mathcal M}_{t_{1},t_{2}}=\int_{t_{1}}^{t_{2}} g_{H}(\dot{\gamma},\dot{\gamma})dt
\ee
which shows that critical points of the later functional are geodesics in the complex hyperbolic plane up to an afine transformation, namely $\gamma(t)=\xi(\alpha t+\beta)$ with $\xi(s)$ a geodesic parametrized by arc length $s$.  Moreover, because 
\be\label{MGE}
\tilde{\mathcal M}_{t_{1},t_{2}}(\gamma)\geq \frac{{\rm length}_{\chp}^{2}(\gamma)}{t_{2}-t_{1}}\geq \frac{{\rm dist}^{2}_{\chp}(\gamma_{1},\gamma_{2})}{t_{2}-t_{1}}
\ee
global minimizers are exactly those critical points $\gamma(t)=\xi(\alpha t+\beta)$ for which $\xi$ is a length minimizing geodesic between $\gamma(t_{1})$ and $\gamma(t_{2})$.

The following Lemma, which is constructed on the previous observation, is central to prove the Lemmas \ref{teo}-\ref{unic}. The proofs are given afterwards.   

\begin{lemma} \label{lema1ax}

\
\begin{enumerate}
\item[\rm (1)] There exists a unique smooth minimizer ${\mathcal D}=(\sigma,\omega,\psi,\chi)$ for the functional $\mathcal M_{\theta_1,\theta_2}$ with given Dirichlet boundary conditions ${\mathcal D}(\theta_1)$, ${\mathcal D}(\theta_2)$. Moreover $\gamma(t)=(\eta,\omega,\psi,\chi)(t)=\xi(\alpha t+\beta)$ where $\xi(s)$ is a geodesic of $\mathbb{H}^{2}_{\mathbb C}$ parametrized by arc-length $s$. 

\item[\rm (2)] The general expression for the unique minimizer of $\tilde{\mathcal M}_{t_1,t_{2}=-t_{1}}$ with centered boundary data $(\eta,\omega,\psi,\chi)|_{t_{1}}=(\eta,-\omega,-\psi,-\chi)|_{t_{2}=-t_{1}}$, $\chi(t_{1})=0, \omega(t_{1})\neq 0$, is given by
\begin{align}\label{u1}
&\eta=\left(\frac{1}{2}c_{5}^2+\frac{1}{2}\sqrt{4c_{1}^{2}+c_{5}^{4}}\cosh \alpha t\right)^{-1},\\
\label{fff} &\chi+i\psi=c_4e^ {if}+\frac{c_3}{c_1},\text{ with } f=-2\arctan\left(\frac{\sqrt{4c_1^2+c_5^4}-c_5^2}{2c_1}\tanh \frac{\alpha t}{2}\right),\\
&\label{fffu} \omega= -\frac{\alpha}{2c_{1}}\sqrt{1-c_{1}^{2}\eta^{2}-4c_{1}^{2}c_{4}^{2}\eta}-2\alpha \psi\chi+\frac{\alpha c_{3}}{c_{1}}\psi
\end{align}
where $ c_1\neq 0$, $c_3$, $c_4$ and $c_5 =2c_1 c_4$ are constants uniquely determined by the boundary conditions at $t_1$ and $-t_1$.
\item[\rm (3)] For any positive sequence $\theta^{i}\rightarrow 0$ and sequence $\{(\sigma_{1}^{i},\omega_{1}^{i},\chi_{1}^{i},\psi_{1}^{i})\}$, such that
\begin{gather} 
\lim \sigma_{1}^{i}=\sigma_{l}\neq\infty,    \qquad    \lim \omega_{1}^{i}=-4J,\qquad  \lim \psi_{1}^{i}=-Q:=-Q_{\rm E}\text{ and,}\\ \chi_{1}^{i}=0, \omega^{i}_{1}\neq 0, \text{ for all $i$,}     
\end{gather}
the unique minimizer ${\mathcal{D}}^{i}(\theta)=(\sigma^{i},\omega^{i},\chi^{i},\psi^{i})$ of ${\mathcal{M}}_{\theta^{i},\pi-\theta^{i}}$ with boundary data
\begin{gather}\label{BD}
{\mathcal{D}}^{i}(\theta^{i})=(\sigma_{1}^{i},\omega_{1}^{i},\chi_{1}^{i},\psi_{1}^{i}),\ 
{\mathcal{D}}^{i}(\pi-\theta^{i})=(\sigma_{1}^{i},-\omega_{1}^{i},-\chi_{1}^{i},-\psi_{1}^{i})
\end{gather}
has
\begin{equation}\label{X79}
\sigma^{i}(\theta)=-\ln \bigg[\frac{1}{2}(c_{5}^{i})^{2}\sin^{2}\theta+\frac{1}{2} \sqrt{4(c_{1}^{i})^{2}+(c_{5}^{i})^{4}}\sin^{2}\theta \cosh \alpha^{i}t\bigg],
\end{equation}
for constants $c_{1}^{i},c_{5}^{i},\alpha^{i}$ and where, as before, $t=\ln \tan \frac{\theta}{2}$. Moreover as $\theta^i\to0$
\begin{equation}\label{X80}
\lim \alpha^{i}= 2 ,\qquad 
\lim\  c_{5}^{i}=\frac{-Q}{\sqrt{Q^{4}+4J^{2}}},\qquad
\lim\ c_{1}^{i}=\frac{J}{Q^{4}+4J^{2}},
\end{equation}
and if we write 
\begin{equation}
\sigma_{l}=\frac{1}{2}\ln (Q^{4}+4J^{2})+\Gamma,
\end{equation}
then 
\begin{equation}\label{X82}
\lim\ \left(\frac{\theta^{i}}{2}\right)^{\alpha^{i}-2}=e^{\Gamma}.
\end{equation}
\end{enumerate}
\end{lemma}
\begin{proof}\

(1) As they differ in a constant, a global minimizer for ${\mathcal M}_{\theta_{1},\theta_{2}}$ is a global minimizer for $\tilde{\mathcal M}_{t_{1},t_{2}}$. Moreover, as explained above, the  later is of the form $\gamma(t)=\xi(\alpha t+\beta)$ with $\xi(s)$ a geodesic parametrized by arc-length $s$ and realizing the distance ${\rm dist}_{\mathbb{H}_{\mathbb{C}}^{2}}(\gamma(t_{1}),\gamma(t_{2}))$ between $\gamma(t_{1})$ and $\gamma(t_{2})$. If $\xi(s=0)=\gamma(t_{1})$ and $\xi(s={\rm dist}_{\mathbb{H}^{2}_{\mathbb{C}}})=\gamma(t_{2})$ (which can always be chosen to be by redefining $s$ if necessary) then $\alpha$ and $\beta$ are unique and determined by $t_{1},t_{2}$ and ${\rm dist}_{\mathbb{H}_{\mathbb{C}}^{2}}(\gamma(t_{1}),\gamma(t_{2}))$. But because $\mathbb{H}^{2}_{\mathbb{C}}$ has negative sectional curvature and is simply-connected, then between two different points $\gamma(t_{1})$ and $\gamma(t_{2})$, there is always a unique minimizing geodesics $\xi(s)$, with $\xi(0)=\gamma(t_{1})$ and $\xi({\rm dist}_{\mathbb{H}_{\mathbb{C}}^{2}}(\gamma(t_{1}),\gamma(t_{2})))=\gamma(t_{2})$. It follows that the global minimizer of ${\mathcal M}_{\theta_{1},\theta_{2}}$, exists, is unique, and has the desired form.

(2) We describe how to obtain a general expression for the unique minimizers of $\tilde{\mathcal M}_{t_{1},t_{2}=-t_{1}}$ whose boundary data satisfy 
\be\label{PR}
(\sigma(t_{1}),\omega(t_{1}),\psi(t_{1}),\chi(t_{1}))=(\sigma(t_{2}),-\omega(t_{2}),-\psi(t_{2}),-\chi(t_{2})),\ \chi(t_{1})=0,
\ee
with  $\omega(t_{1})\neq 0$. The Euler-Lagrange equations for $\tilde{\mathcal M}_{t_{1},t_{2}}$ are integrable and the first integrals can be obtained as conserved quantities of the form $g_{H}(X,\dot\gamma)$ which arise from Killing fields $X^a$ for $g_{H}$. The Killing fields we will use are
\begin{equation}
X_{1}=\partial_{\omega},\qquad X_{2}=-2\psi\partial_{\omega}+\partial_{\chi},\qquad X_{3}=2\chi\partial_{\omega}+\partial_{\psi}
\end{equation}
The corresponding conserved quantities can be combined to give
\begin{equation}\label{EQchi}
\frac{\dot\omega+2\chi\dot\psi-2\psi\dot\chi}{\eta^2}=\alpha c_{1},\qquad \alpha c_{1}\psi-\frac{\dot\chi}{\eta}=\alpha c_{2},\qquad \alpha c_{1}\chi+\frac{\dot\psi}{\eta}=\alpha c_{3}
\end{equation}
where $c_{1},c_{2}$ and $c_{3}$ are constants and we have inserted explicitly the  (positive) constant $\alpha$ (introduced in item (1) before).  Note that $c_{1}\neq 0$ for if $c_{1}=0$ then (\ref{PR}) and (\ref{EQchi}) imply $\omega$ identically zero which contradicts $\omega(t_{1})\neq 0$. To obtain the equation for $\eta$ (or equivalently, for $\sigma$) we use $g_{H}(\dot\gamma,\dot\gamma)=\alpha^{2}$ and thus
\be\label{EQu0}
\frac{\dot\eta^{2}}{\eta^2}+\frac{(\dot\omega+2\chi\dot{\psi}-2\psi\dot{\chi})^{2}}{\eta^2}+4\frac{\dot\chi^{2}+\dot\psi^{2}}{\eta}=\alpha^2
\ee
Equations \eqref{EQchi}-\eqref{EQu0} are indeed equivalent to the equations of motion obtained from the variation of (\ref{Energy}), cf. \eqref{ech1}-\eqref{echf}.
These equations can be further simplified by using an important property of the variables $(\psi,\chi)$. By making $\bar{\psi}=\psi-c_{2}/c_{1}$ and $\bar{\chi}=\chi-c_{3}/c_{1}$,  the second and third equations in (\ref{EQchi}) reduce to
\begin{equation}\label{EQchi2}
\alpha c_{1}\bar\psi-\frac{\dot{\bar\chi}}{\eta}=0,\qquad \alpha c_{1}\bar\chi+\frac{\dot{\bar\psi}}{\eta}=0.
\end{equation}
Multiplying these two equations respectively by $\bar{\chi}$ and $\bar{\psi}$ and substracting one from the other we obtain $\bar{\chi} \dot{\bar{\chi}}+\bar{\psi}\dot{\bar{\psi}}=0$ which implies $\bar{\chi}^{2}+\bar{\psi}^{2}=c^{2}_{4}$  where $c_{4}$ is a constant. We write
\begin{equation}\label{f}
\bar{\chi}+i\bar{\psi}=c_{4}e^{if}, \qquad \mbox{with}\qquad \dot f=-c_{1}\alpha \eta.
\end{equation}
Now, since $\dot{\bar{\chi}}=\dot\chi$ and $\dot{\bar{\psi}}=\dot\psi$, then 
\be\label{chipsis}
\dot{\chi}^{2}+\dot{\psi}^{2}=\alpha^ 2c^{2}_{4}c_{1}^{2}\eta^2.
\ee 
We use equations (\ref{EQchi}) and (\ref{chipsis}) to rewrite \eqref{EQu0} as
\be\label{EQuI}
\frac{\dot{\eta}^2}{\eta^2}+\alpha^ 2c_{1}^{2}\eta^2+4\alpha^ 2 c_{1}^{2}c_{4}^{2}\eta=\alpha^2.
\ee
We now solve equation \eqref{EQuI} for $\eta$ and use $\eta(t_{1})=\eta(t_{2}),\ t_{1}=-t_{2}$, and find \eqref{u1}
with $c_5^2:=4c_1^ 2c_4^2$.

Now we solve for $\psi,\chi$. Note that in order to have $\chi(t_1)=\chi(t_2)=0$ and $\psi(t_1)=-\psi(t_2)$ and at the same time $\bar{\psi}^{2}+\bar{\chi}^{2}=c_{4}^{2}$ the only possibility is to have 
$c_2=0$ and therefore
\be\label{chipsi}
\chi+i\psi=c_{4}e^{if}+\frac{c_{3}}{c_1}. 
\ee
We obtain $f$ by integrating the second equation in line \eqref{f}, using \eqref{u1} and $f(t_1)=-f(t_2)$, to find \eqref{fff}
To find $\psi$ and $\chi$ use (\ref{chipsi}) where $c_{3}$ is adjusted from $c_{1},c_{5},\alpha$ to have $\chi(t_{1})=\chi(t_{2})=0$.
To find $\omega$ on the other hand one could use the expression for $\dot{\omega}$ in (\ref{EQchi}) and integrate. However a direct and simple expression for $\omega$ arises when using the conserved quantity associated to the Killing field
$X_{4}=2\eta\partial_{\eta}+2 \omega\partial_{\omega}+\chi\partial_{\chi}+\psi\partial_{\psi}$. Explicitly 
\begin{equation}\label{x4}
 g_{H}(X_{4},\dot\gamma)=\frac{\dot\eta}{2\eta}+ \frac{c_{1}\alpha}{2} \omega +\frac{\chi\dot\chi+\psi\dot\psi}{\eta}=c
\end{equation} 
Noting that the above expression is antisymmetric in $t$ around $t=0$ we deduce that the constant $c$ is zero. Then, from \eqref{x4} one obtains a direct expression for $\omega$.  The expressions that one obtains for $\omega,\psi,\chi$ in this form are somehow crude, but serve well to the purposes of the proof of (3).  Summarizing, given $c_{1} \neq ,c_{5},\alpha,\theta_{1}$ one can associate, following the construction  above, a unique solution $\gamma(t)$ satisfying (\ref{PR}) with $-t_{2}=t_{1}=\ln \tan \frac{\theta_{1}}{2}$.  

(3) A priori, to prove item (3) one could calculate the constants $(c_{1}^{i},c_{5}^{i},\alpha^{i})$ from the prescribed boundary data at $\theta^{i},\pi-\theta^{i}$ and prove from them, by a direct calculation, the conclusions (\ref{X80}) and (\ref{X82}).  Unfortunately such procedure is a computational nuisance. For this reason we follow an alternative argument. 
Given $\theta_{1} >0$ consider the map $F_{\theta_{1}}:\mathbb{R}^{3}\setminus \{ y=0\}\rightarrow \mathbb{R}^{3}$ that to $(\Gamma,c_{1},c_{5})$ associates the boundary values $(\sigma(\theta_{1}),\omega(\theta_{1}),\psi(\theta_{1})$ of the solution $(\sigma,\omega,\psi,\chi)$ found from the constants $(c_{1},c_{5},\alpha,\theta_{1})$ where $\alpha$ is given by
\be
\alpha=2+\frac{\Gamma}{\ln \theta_{1}/2}
\ee
Then, if we let $\theta_{1}\rightarrow 0$, from (\ref{u1}) and the limit $\lim_{\theta_{1}\rightarrow 0} \sin^{2}\theta_{1} \cosh (\alpha\ln \frac{\theta_{1}}{2})=\frac{2}{\Gamma}$, we obtain $\lim \sigma(\theta_{1})=\Gamma-\frac{1}{2}\ln (4c_{1}^{2}+c_{5}^{4})$. Next, using that $\psi(\theta_{1})=c_{4}\sin f(\theta_{1})$ and that $f(\theta_{1})\rightarrow \arctan 2c_{1}/c_{5}^{2}$ we find $\lim\psi(\theta_1)=\frac{ c_{5}}{\sqrt{4c_{1}^{2}+c_{5}^{4}}}$. Finally, from (\ref{x4}), (\ref{EQu0}) and the expression
\be
\frac{\psi\dot\psi}{\eta}\bigg|_{\theta_{1}}=-c_{1} \alpha c_{4}^{2}\sin f\cos f\bigg|_{\theta{1}}\rightarrow -\frac{c_{5}^{4}}{4c_{1}^{2}+c_{5}^{4}}
\ee
we get $\lim\omega(\theta_1)=\frac{-4c_{1}}{4c_{1}^{2}+c_{5}^{4}}$. This shows that the maps $F_{\theta_{1}}$ converge uniformly on any compact set to a map $F_{0}$ given by
 \be
F_{0}(\Gamma,c_{1},c_{5})=(\Gamma-\frac{1}{2}\ln (4c_{1}^{2}+c_{5}^{4}),-\frac{4c_{1}}{4c_{1}^{2}+c_{5}^{4}},\frac{c_{5}}{4c_{1}^{2}+c_{5}^{4}}).
\ee
Moreover the map $F_{0}$  extends to a diffeomorphism from $\mathbb{R}^{3}\setminus (\{y=0\}\cap \{z=0\})$ into $\mathbb{R}^{3}\setminus (\{y=0\}\cap \{z=0\})$. A close inspection of the limits above shows that the maps $F_{\theta}$ do extend smoothly too.
We note now that given the values $(\sigma_{l},J,Q_{\rm E}=Q,0)$  (with either $J$ or $Q$ non-zero) if we take 
\be
(\Gamma^{\infty},c_{1}^{\infty},c_{5}^{\infty})=\left(\sigma_{l}-\frac{1}{2}\ln \big[ 4J^{2}+Q^{4}\big], \frac{J}{4J^{2}+Q^{4}},\frac{-Q}{\sqrt{4J^{2}+Q^{4}}}\right)
\ee
then $F_{0}(\Gamma^{\infty},c_{1}^{\infty},c_{5}^{\infty})=(\sigma_{l},J,Q)$. It follows therefore from the above argument that given $(\sigma_{l},J,Q)$  (with either $J$ or $Q$ non-zero) and sequences $\{\theta^{i}\rightarrow 0\}$ and $\{(\sigma^{i},\omega^{i},\psi^{i},\chi^{i})\}$ as in the hypothesis of (3), then there is a sequence $\{(\Gamma^{i},c_{1}^{i},c_{5}^{i})\}$ with limit $(\Gamma^{\infty},c_{1}^{\infty},c_{5}^{\infty})$ such that,  $F(\Gamma^{i},c_{1}^{i},c_{5}^{i})=(\sigma^{i},\omega^{i},\psi^{i})$ and therefore the unique minimizer of ${\mathcal M}_{\theta_{1}^{i},\pi-\theta_{1}^{i}}$ with boundary data (\ref{BD}) is the unique solution constructed out of $(c_{1}^{i},c_{5}^{i},\alpha^{i},\theta_{1}^{i})$ where $\alpha^{i}=2+\frac{\Gamma^{i}}{\ln \theta_{1}^{i}/2}$ . The expressions (\ref{X79}), (\ref{X80}) and (\ref{X82}) are readily checked. This finishes (3) and the proof of the Lemma.
\end{proof}

\begin{proof} \textit{(of Lemma \ref{teo})} Consider any sequence $\{\theta^{i}_{1}\downarrow 0\}$. Now we divide the interval $(0,\pi)$ in three parts, and write, for the set $\mathcal D$
\begin{equation}
{\mathcal{M}}({\mathcal D})={\mathcal{M}}_{0,\theta^{i}_{1}}+{\mathcal{M}}_{\pi-\theta^{i}_{1},\pi}+{\mathcal{M}}_{\theta^{i}_{1},\pi-\theta^{i}_{1}}.
\end{equation}
Then we recall the relation between $\tilde{\mathcal M}_{\theta^{i}_{1},\pi-\theta^{i}_{1}}$ and $\mathcal M_{\theta^{i}_{1},\pi-\theta^{i}_{1}}$ which is
\begin{equation}\label{rel}
 \tilde{\mathcal M}_{\theta^{i}_{1},\pi-\theta^{i}_{1}}(\gamma)=\mathcal M_{\theta^{i}_{1},\pi-\theta^{i}_{1}}({\mathcal D})+4\sigma\cos\theta\bigg|^{\pi-\theta^{i}_{1}}_{\theta^{i}_{1}}+4\cos\theta\bigg|^{\pi-\theta^{i}_{1}}_{\theta^{i}_{1}}+4\ln\tan\frac{\theta}{2}\bigg|^{\pi-\theta^{i}_{1}}_{\theta^{i}_{1}}
\end{equation}
where of course $\gamma$ represents the same data as ${\mathcal D}$.  Using this we would like to get a sharp estimation from below to $\tilde{\mathcal M}_{\theta_{1}^{i},\pi-\theta^{i}_{1}}$. For this we proceed as follows. For every $i$ consider two points in $\chp$, denoted by $\bar{\gamma}_{\theta^{i}_{1}}, \bar{\gamma}_{\pi-\theta^{i}_{1}}$ and given by 
\be
\bar{\gamma}_{\theta^{i}_{1}}=(\eta(\theta^{i}_{1}),-4J,-Q_{\rm E},-Q_{\rm M}=0),\ \bar{\gamma}_{\pi-\theta^{i}_{1}}=(\eta(\theta^{i}_{1}),+4J,+Q_{\rm E},+Q_{\rm M}=0)
\ee
 if $J\neq 0$, while if $J=0$ then we replace $4J$ in the above expressions by $\hat{\omega}^{i}_{1}$ tending to zero sufficiently fast (see below). This is because below we will need to use Lemma \ref{lema1ax} (3), for the minimizers with boundary data  $\bar{\gamma}_{\theta^{i}_{1}}$ and $\bar{\gamma}_{\pi-\theta^{i}_{1}}$, but Lemma \ref{lema1ax} requires non-zero boundary values for $\omega$. From the regularity at the poles one easily deduces that  (if $J\neq 0$, or if $\hat{\omega}^{i}_{1}$ goes to zero sufficiently fast) ${\rm dist}_{\chp}(\bar{\gamma}_{\theta^{i}_{1}}, \gamma(\theta^{i}_{1}))\rightarrow 0$ (see expression (\ref{dist})) and similarly for ${\rm dist}_{\chp}(\bar{\gamma}_{\pi - \theta^{i}_{1}},\gamma(\pi - \theta^{i}_{1}))$. Consider now any another sequence $\{\theta^{i}_{2}\downarrow 0\}$, such that $\theta^{i}_{2}<\theta^{i}_{1}$ for every $i$ and 
\begin{gather}
\label{RR} \lim \frac{{\rm dist}^{2}_{\chp}(\gamma(\theta^{i}_{1}),\bar{\gamma}_{\theta^{i}_{1}})}{t(\theta^{i}_{2})-t(\theta^{i}_{1})}=\lim \frac{{\rm dist}^{2}_{\chp}(\gamma(\pi-\theta^{i}_{1}),\bar{\gamma}_{\pi-\theta^{i}_{1}})}{t(\theta^{i}_{2})-t(\theta^{i}_{1})}=0,\\
\label{RRR} \lim t(\theta^{i}_{1})-t(\theta^{i}_{2})=0, 
\end{gather}
with, again, $t(\theta)=\ln \tan \frac{\theta}{2}$. Finally consider the curve in $\chp$, denoted by $\bar{\gamma}^{i}$, starting at $\bar{\gamma}_{\theta^{i}_{1}}$ and ending at $\bar{\gamma}_{\pi-\theta^{i}_{1}}$ defined as
\begin{enumerate}
\item the minimizer of $\tilde{\mathcal M}_{t(\theta_{2}^{i}),t(\theta_{1}^{i})}$, with boundary data $\bar{\gamma}_{\theta^{i}_{1}}$, $\gamma(t(\theta^{i}_{1}))$, if $t\in [t(\theta^{i}_{2}),t(\theta_{1}^{i})]$, 
\item $\gamma(t)$ if $t\in [t(\theta^{i}_{1}),t(\pi-\theta^{i}_{1})]$,
\item the minimizer of $\tilde{\mathcal M}_{t(\pi-\theta_{1}^{i}),t(\pi- \theta_{2}^{i})}$, with boundary data $\gamma(t(\pi-\theta^{i}_{1}))$, $\bar{\gamma}_{\pi-\theta^{i}_{1}}$, if $t\in [t(\pi-\theta^{i}_{1}),t(\pi-\theta_{2}^{i})]$
\end{enumerate}
By (\ref{MGE}) we can write 
\be\label{RII}
\tilde{\mathcal M}_{t(\theta^{i}_{2}),t(\pi-\theta^{i}_{2})}(\bar{\gamma}^{i})\geq \alpha_{i}^{2}(t(\pi-\theta^{i}_{2})-t(\theta^{i}_{2}))=-2\alpha_{i}^{2}\ln \tan \frac{\theta^{i}_{2}}{2}
\ee
where $\alpha_{i}$ is the constant associated to the minimizer of $\tilde{\mathcal M}_{t(\theta^{i}_{2}),t(\pi-\theta^{i}_{2})}$ with boundary data $\bar{\gamma}_{\theta^{i}_{1}},\ \bar{\gamma}_{\pi-\theta^{i}_{1}}$, as in Lemma \ref{lema1ax}. By (\ref{MGE}), (\ref{RR}) and (\ref{RRR}) we have 
\be\label{RIII}
\lim \tilde{\mathcal M}_{t(\theta_{2}^{i}),t(\theta_{1}^{i})}(\bar{\gamma}^{i})=\lim \tilde{\mathcal M}_{t(\pi-\theta_{1}^{i}),t(\pi-\theta_{2}^{i})}(\bar{\gamma}^{i})=0
\ee
and finally, of course,
\be\label{RIV}
\tilde{\mathcal M}_{t(\theta^{i}_{2}),t(\pi-\theta^{i}_{2})}(\bar{\gamma}^{i}) = \tilde{\mathcal M}_{t(\theta_{2}^{i}),t(\theta_{1}^{i})}(\bar{\gamma}^{i})+ \tilde{\mathcal M}_{t(\pi-\theta_{1}^{i}),t(\pi-\theta_{2}^{i})}(\bar{\gamma}^{i}) + \tilde{\mathcal M}_{t(\theta^{i}_{1}),t(\pi-\theta^{i}_{1})}(\gamma)
\ee
Collecting (\ref{RIV}) and (\ref{rel}) together with the information (\ref{RRR}), (\ref{RII}) and (\ref{RIII}) we obtain,
\be
\lim {\mathcal M}_{\theta_{1}^{i},\pi-\theta_{1}^{i}}({\mathcal D})\geq \lim \bigg[-2\big( \alpha_{i}^{2}-4\big)\ln \tan \frac{\theta_{1}^{i}}{2}-4\sigma\cos\theta\bigg|_{\theta^{i}_{1}}^{\pi-\theta^{i}_{1}}+8\bigg]
\ee
But as $\sigma(\theta_{1}^{i})\rightarrow \sigma_{l}=\frac{1}{2}\ln (4J^{2}+Q^{4})+\Gamma$ with $\lim (\theta^{i}_{1}/2)^{\alpha_{i}^{2}- 4}=e^{\Gamma}$ we obtain, after a cancelation,
\begin{equation}\label{ec13}
{\mathcal M}({\mathcal D})=\lim \mathcal M_{\theta_{1}^{i},\pi-\theta_{1}^{i}}\geq 4\ln(Q^4+4J^2)+8=\mathcal M^0.
\end{equation}
\end{proof}

We present now the proof of Lemma \ref{unic}. This is achieved by making use of the explicit expression for the minimizers of the functional $\mathcal M$ with given boundary conditions found above.

\begin{proof} \textit{(of Lemma \ref{unic})}.
We know that any critical point of ${\mathcal M}$ is represented in terms of a geodesic $\gamma=(\eta,\omega,\psi,\chi)$ of $\chp$. 
Regularity implies $\lim \eta'/\eta=2$ as $\theta$ tends to $0$ or $\pi$. This implies from (\ref{EQuI}) that $\alpha=2$. On the other hand, the boundary data implies  
\begin{equation}\label{c5f} 
c_{1}= \frac{J}{Q^{4}+4J^{2}},\qquad 
c_{5}= \frac{ -Q}{\sqrt{Q^{4}+4J^{2}}},
\end{equation}
 Then, manipulating (\ref{fff}) while using these values for $c_{1}$ and $c_{5}$ one obtains 
\be
f=\arctan \frac{2c\cos\theta}{1-c^{2}\cos^{2}\theta}
\ee
with $c=2J/(\sqrt{4J^{2}+Q^{4}}+Q^{2})$. Plugin this expression in ``the law of the two arcs" (\ref{LTA}) gives (\ref{EKN2})-(\ref{EKN4}) with $Q_{\rm M}=0$. The equation (\ref{EKN1}) is obtained from (\ref{u1}).
\end{proof}

\appendix

\section{Appendix: linking global and quasilocal axisymmetric inequalities}

In this somehow more informal appendix we want to  show that there might exist a link between the  AJQ and MJQ inequalities (see equations \eqref{ine} and \eqref{MJQ}). The MJQ inequality \eqref{MJQ} is a global manifestation of the constraints (in the maximal spatial gauge) and in this sense it is a global inequality requiring knowledge of the system as  a whole. The AJQ inequality \eqref{ine} is instead of a quasilocal nature and does not require global information. Despite of the different realms in which they manifest, they seem to be closely related. The link that we shall establish could be of help to prove quasilocal inequalities in systems other than Einstein-Maxwell, for which a global three dimensional mass functional is shown to exist. More concretely we will point out a relation among (\ref{ine}) and (\ref{MJQ}) by linking the two-dimensional energy functional ${\mathcal{M}}$ given in (\ref{defM}) and a three-dimensional energy functional given in \cite{ChruscielCosta09}-\cite{Costa09} (see below), and whose minimization properties lead to the AJQ and MJQ inequalities respectively.

Let us first put in parallel how one obtains the MJQ and AJQ inequalities from suitable functionals.

The inequality (\ref{MJQ}) has been established in \cite{ChruscielCosta09}-\cite{Costa09} following a similar argument as in \cite{Dainmass08}. The black hole configuration on which (\ref{MJQ}) has been proved is that of an initial datum with two asymptotically flat ends, where $m,J,Q_{\rm E},Q_{\rm M}$ in the inequality (\ref{MJQ}) refers to the mass, angular momentum and charges of a selected end. 
The rationale behind the proof of (\ref{MJQ}) is the following. One introduces a three-dimensional functional $M$, defined on such configurations and bounding the mass (of the selected end) from below, i.e. $m\geq M$. Moreover one has $M\geq M_{0}$ where $M_{0}$ is the infimum of $M$ among those configurations having (for the selected end) $J,Q_{\rm E}$ and $Q_{\rm M}$ fixed. Moreover $M_{0}$ is achieved by the extreme Kerr-Newman solution. The inequalities $m\geq M$ and $M\geq M_{0}$ together with the explicit expression for $M_{0}$ give (\ref{MJQ}). 

On the other hand the rationale behind (\ref{ine}) that we have developed in the previous sections was the following. We introduced a functional ${\mathcal{M}}$ defined on a certain surface (i.e. stable MOTS or stable minimal surface over a maximal slice) and bounding its area from below, more precisely by $A\geq 4\pi e^{({\mathcal{M}}-8)/8}$. Then we showed that the extreme Kerr-Newman sphere realizes the absolute minimum of ${\mathcal{M}}$  among all configurations having $J$, $Q_{\rm E}$ and $Q_{\rm M}$ fixed. Denoting by ${\mathcal{M}}_{0}$ the value of ${\mathcal{M}}$ at the extreme Kerr-Newman sphere we get $A\geq 4\pi e^{({\mathcal{M}}_{0}-8)/8}$ which gives (\ref{ine}). 

It is clear that the two procedures described above are formally similar and we will see that, although both can be carried out without any reference to one another, they are indeed remarkably related.
More precisely we state that the inequality $m\geq M\geq M_{0}$ implies that the extreme Kerr-Newman sphere is a critical point of ${\mathcal{M}}$ (even more, it can be shown from that, that the extreme Kerr-Newman sphere is a local minimum for ${\mathcal{M}}$). However, we do not know at the moment whether the fact that the extreme Kerr-Newman sphere is a global minimizer of ${\mathcal{M}}$  can be established solely from the inequality $m\geq M\geq M_{0}$. This gives a partial connection in the form MJQ $\Rightarrow$ AJQ. In the other direction, namely AJQ $\Rightarrow $ MJQ, we note that with the help of the Penrose inequality $A\leq 16\pi m^{2}$ one obtains for outermost minimal surfaces
\be
m^{2} \geq \frac{A}{16\pi}\geq \frac{4J^{2}+Q^{4}}{4},
\ee

\n which is an inequality slightly worst than (\ref{MJQ}). Despite of these interesting relations many issues on to the link between the inequalities still remain in shadows. 


\vs
 In order to see how the first implication shows up, we begin by defining the three-dimensional potentials $\bar{\mathcal D}=(\bar{\sigma},\bar{\omega},\bar{\psi},\bar{\chi})$ over maximal electrovacuum initial data $(\Sigma, h, K, F)$ and then introduce the functional $M$ together with a crucial minimizing property. We follow \cite{Costa09} on this construction.

We write the spatial metric on the maximal initial datum in the form
\be\label{Metrich}
h=e^{\bar{\sigma}+2\bar{q}}(\frac{d\bar{r}^{2}}{\bar{r}^{2}}+d\theta^{2})+e^{\bar{\sigma}}\sin^{2}\theta (d\varphi +v_{\bar r}d\bar r+v_\theta d\theta)^{2}.
\ee
where $v_{\bar r}, v_\theta$ are functions of $\bar r, \theta$, which spans over $\mathbb{R}^{3}\setminus \{0\}$ and defines $\bar{\sigma}$. For the electromagnetic fields we have the relations 
\begin{gather}
 \label{EB}  E_{a}=F_{ab}n^{b},\qquad B_{a}={}^*\!F_{ab}n^{b},\\
 \label{PC}  \partial_{a}\bar{\chi}:=F_{ab} \eta^{b},\qquad  \partial_{a}\bar{\psi}:={}^*\!F_{ab} \eta^{b},
\end{gather}
where  $n^a$ is the unit normal to $\Sigma$. These expressions define the potentials $\bar{\psi},\ \bar{\chi}$. Finally a potential $\bar{\omega}$ is defined through 
\be
D_{a}\bar{\omega}+2\bar{\chi} D_{a}\bar{\psi}-2\bar{\psi} D_{a}\bar{\chi}:=2\epsilon_{abc}K^{b}_{\ d} \eta^{c}\eta^{d}.
\ee
Observe that because the norm of the axial Killing vector $\eta^a$ is null over the axis then the differential of the potentials $(\bar{\omega},\bar{\psi},\bar{\chi})$ at the axis are zero and therefore their values remain constant all along them. As they are defined up to a constant one can take them to be ``centered", namely, $\bar\omega|_{\theta=\pi}=-\bar\omega|_{\theta=0}$, $\bar\psi|_{\theta=\pi}=-\bar\psi|_{\theta=0}$ and $\bar\chi|_{\theta=\pi}=-\bar\chi|_{\theta=0}$.  

It is interesting and illustrative to see the relation between the potentials $\bar{\mathcal{D}}_{0}=(\bar{\sigma}_{0},\bar{\omega}_{0},\bar{\psi}_{0},\bar{\chi}_{0})$ corresponding to the extreme Kerr-Newman solution over, say, the slice $\{t=0\}$, and the  potentials ${\mathcal D}_{0}=(\sigma_{0},\omega_{0},\psi_{0},\chi_{0})$ defining the extreme Kerr-Newman sphere (\ref{EKN1})-(\ref{EKN4}). The explicit form of the 3-dimensional potentials $\bar{\mathcal D}_{0}$ can be found in \cite{Carter72}  (see pages 197-204) and we have
\begin{gather}
\bar{\sigma}_{0}=\ln \frac{(\bar{r}^{2}-Q^{2}+2m_{0}(\bar{r}+m_{0}))^{2}-\bar{r}^{2}a_{0}^{2}\sin^{2}\theta}{\Sigma},\\ 
2(\bar{\sigma}_{0}+\bar{q}_{0})=\ln \bigg[ (\bar{r}^{2}-Q^{2}+2m_{0}(\bar{r}+m_{0}))^{2}-\bar{r}^{2}a_{0}^{2}\sin^{2}\theta\bigg]
\end{gather}
where $\bar{r}=r-r_{H}=r-m_{0}$. From this we obtain
\begin{gather}
\label{SQ1} \lim_{\bar{r}\rightarrow 0} \bar{\sigma}_{0}(\bar{r},\theta,\varphi)=\sigma_{0}(\theta,\varphi),\\
\label{SQ2} \lim_{\bar{r}\rightarrow 0} 2(\bar{\sigma}_{0}(\bar{r},\theta,\varphi)+\bar{q}_{0}(\bar{r},\theta,\varphi))=\ln 4J^{2}+Q^{4}=\ln \frac{A^{2}}{16\pi^{2}}=2c
\end{gather} 
where $\sigma_{0}$ is given by (\ref{EKN1}) and (as before) $A=4\pi e^{c}$. Together with (\ref{Metrich}) this shows that the $\{(\theta,\varphi)\}$ coordinates on the spheres $\{\bar{r}=\bar{r}_{1}\}$ become, as $\bar{r}_{1}\rightarrow 0$, the unique ones for which the induced metric (over $\{\bar{r}=\bar{r}_{1}\}$) is expressed in the form (\ref{SDE}).
Moreover from (\ref{EB})-(\ref{PC}) it is deduced that over any sphere $\{\bar{r}=\bar{r}_{1}\}$ it holds 
\begin{gather}\label{EBP}
E_{\perp}=\frac{e^{\bar{\sigma}_{0}+\bar{q}_{0}}\partial_{\theta}\bar{\psi}_{0}}{\sin\theta},\qquad B_{\perp}=\frac{e^{\bar{\sigma}_{0}+\bar{q}_{0}}\partial_{\theta}\bar{\chi}_{0}}{\sin\theta}
\end{gather}
As was explained in Section \ref{seckerr},  $E_{\perp}$ and $B_{\perp}$ converge as $\bar{r}_{1}\rightarrow 0$ to those given by (\ref{K3}) and (\ref{K4}) respectively, that is those of the extreme Kerr-Newman sphere. From this, (\ref{defpsichi}) and (\ref{SQ2}), we deduce that the limit of the potentials $\bar{\psi}$ and $\bar{\chi}$ over the spheres $\{\bar{r}=\bar{r}_{1}\}$ converge to (\ref{EKN3}) and (\ref{EKN4}) respectively, that is
\begin{equation}
\lim_{\bar{r}\rightarrow 0} \bar{\psi}(\bar{r},\theta,\varphi)=\psi_{0}(\theta,\varphi),\qquad \lim_{\bar{r}\rightarrow 0} \bar{\chi}(\bar{r},\theta,\varphi)=\chi_{0}(\theta,\varphi)
\end{equation}
The same property is also obtained for $\bar{\omega}_{0}$
\be\label{LPO4}
\lim_{\bar{r}\rightarrow 0} \bar{\omega}_{0}(\bar{r},\theta,\varphi)=\omega_{0}(\theta,\varphi)
\ee

Given the set of 3-dimensional potentials $\bar{\mathcal D}$, we define, as done in \cite{Costa09}, the energy functional $M$ on $\bar{\mathcal D}$ 
\be\label{COSF}
M=\int_{\mathbb{R}^{3}} 4\bigg(|D\bar{U}|^{2}+\frac{e^{4\bar{U}}}{\rho^{4}}|\frac{D \bar{\omega}/2+\bar{\chi} D\bar{\psi}-\bar{\psi} D\bar{\chi}}{2}|^{2}+\frac{e^{2\bar{U}}}{\rho^{2}}(|D\bar{\chi}|^{2}+|D\bar{\psi}|^{2})\bigg)\ dV_{0}.
\ee
\n where $e^{\bar{\sigma}}=e^{-2\bar{U}}\bar{r}^{2}$, $dV_{0}=\bar{r}^{2}\sin\theta d\bar{r} d\theta d\varphi$, $D$ is the Euclidean differential and the norms are Euclidean norms. 

Although in principle the functional $M$ was introduced on axisymmetric maximal initial data with two asymptotically flat ends, we will consider it acting on more general sets of functions $\bar{\mathcal D}(\bar{\sigma},\bar{\omega},\bar{\psi},\bar{\chi})$ with fixed $J,Q_{\rm E}$ and $Q_{\rm M}$, not necessarily arising from the potentials of an initial state. In this setup a key property of $M$, which is deduced from the arguments in \cite{Costa09} is the following. Let $\bar{\mathcal{D}}=(\bar{\sigma},\bar{v},\bar{\chi},\bar{\psi})$ be a  set that is the extreme Kerr-Newman set outside a compact set in $\mathbb{R}^{3}\setminus \{0\}$. Then $M(\bar{\mathcal{D}})\geq M(\bar{\mathcal{D}}_{0})$ where $\bar{\mathcal{D}}_{0}$ is the set for the extreme Kerr-Newman solution. In other words, the extreme Kerr-Newman set is a minimum of $M$ under variations of $\bar{\mathcal D}_{0}$ of compact support. 

We are ready to explain how to deduce that ${\mathcal D}_{0}$ is a critical point for ${\mathcal M}$ from the properties of $M$.
Let ${\mathcal D}=(\sigma,\omega,\psi,\chi)$ be a set on the sphere $S^{2}$ and let ${\mathcal D}_{0}$ be the extreme Kerr-Newman sphere set, both with the same angular momentum and charges $J$, $Q_{\rm E}$ and $Q_{\rm M}$. Define the set for the functional ${\mathcal M}$
\be
{\mathcal D}_{\lambda}=\lambda({\mathcal D}-{\mathcal D}_{0})
\ee

Let $\xi(x)$ be a real function equal to $1$ for $x\leq 0$, equal to $-x+1$ for $x\in [0,1]$ and equal to $0$ for $x\geq 1$. For every $\epsilon>0$ define $\xi_{\epsilon}(\bar{r})=-1/\bar{r}+1/\epsilon$. Finally consider the data for the functional $M$ 
\be
\bar{\mathcal D}_{\lambda,\epsilon}(\bar{r},\theta,\varphi)=\bar{\mathcal D}_{0}(\bar{r},\theta,\varphi)+ {\mathcal D}_{\lambda}(\theta,\phi)
\ee
A long but otherwise straightforward calculation shows
\be
\frac{d {\mathcal M}({\mathcal D}_{\lambda})}{d\lambda}\bigg|_{\lambda=0}=\lim_{\epsilon\rightarrow 0}\frac{1}{\epsilon}\frac{d M(\bar{\mathcal D}_{\lambda,\epsilon})}{d\lambda}\bigg|_{\lambda=0}
\ee
Now, since $M(\bar{\mathcal D} )\geq M(\bar{\mathcal D_0})$ the right hand side is zero for every $\epsilon>0$, therefore the left hand side is zero and, because ${\mathcal D}$ was arbitrary, we conclude that ${\mathcal D}_{0}$ is a critical point of ${\mathcal M}$.

\section*{Acknowledgments}
We are very grateful to Walter Simon and Sergio Dain for helpful discussions.

\end{document}